\documentclass[a4paper,twocolumn,11pt]{quantumarticle}
\pdfoutput=1
\makeatletter
\providecommand\@afterenddocumenthook{}
\makeatother

\usepackage[utf8]{inputenc}
\usepackage[english]{babel}
\usepackage[T1]{fontenc}

\usepackage{amsmath,amssymb,amsthm,mathtools}

\usepackage{tikz}
\usetikzlibrary{matrix,positioning,arrows.meta}
\usepackage{booktabs}
\usepackage{graphicx}
\usepackage{xcolor}
\usepackage{mdframed}

\usepackage{enumitem}
\usepackage{appendix}
\usepackage[numbers,sort&compress]{natbib}

\usepackage{hyperref}

\graphicspath{{figures/}{./}}

\newmdenv[
  linewidth=0.6pt,
  linecolor=black!50,
  backgroundcolor=black!3,
  innertopmargin=8pt,
  innerbottommargin=8pt,
  innerleftmargin=12pt,
  innerrightmargin=12pt,
  skipabove=8pt,
  skipbelow=8pt,
]{summarybox}

\newtheorem{theorem}{Theorem}[section]
\newtheorem{proposition}[theorem]{Proposition}
\newtheorem{lemma}[theorem]{Lemma}
\newtheorem{corollary}[theorem]{Corollary}
\newtheorem{conjecture}[theorem]{Conjecture}
\theoremstyle{definition}
\newtheorem{definition}[theorem]{Definition}

\theoremstyle{remark}
\newtheorem{remark}[theorem]{Remark}

\DeclareMathOperator{\rt}{rt}
\DeclareMathOperator{\diam}{diam}

\begin{document}

\title{Block Permutation Routing on Ramanujan Hypergraphs for Fault-Tolerant Quantum Computing}
\author{Joshua M. Courtney}
\affiliation{University of Georgia, Department of Physics and Astronomy}
\orcid{0009-0002-1631-4277}
\email{Joshua.Courtney1@uga.edu}
\maketitle

\begin{abstract}
We analyze permutation routing of rigid blocks representing surface code patches of $d_C^2$ atoms on a reconfigurable lattice with hypergraph transformations. 
We show that the {block routing number} $\rt_B(H, s, g) = \Theta(d_C \log N_L)$ for a given hypergraph $H$, code distance $d_C$, $s=d_C^2$,  number of blocks $N_L$, and guard distance $g$. 
The main result uses a spectral analysis of the {quotient graph} $Q(G_{\mathrm{cl}}(H), B)$ (blocks as supervertices) through equitable partition interlacing, showing that the spectral ratio $\beta_Q < 1$ is preserved in the high-connectivity regime $d' > d_C^2(r-1)/(1-\beta)$. 
Negative association of block permutations and congestion bounds are used for random intermediate configurations, and serialization establishes that each quotient routing phase requires $O(d_C)$ physical sub-steps due to the block footprint width.
The lower bound $\rt_B = \Omega(d_C \log N_L)$ follows from combining the spectral lower bound $\Omega(\log N_L)$ on quotient phases with the $\Omega(d_C)$ footprint traversal cost per phase.
We include error model analysis grounded in recent experimental results from Bluvstein et al.~\cite{bluvstein2024logical,bluvstein2026fault,evered2023high} and Google Quantum AI~\cite{acharya2024quantum,google2023suppressing}, syndrome extraction protocols (stop-and-correct, rolling active fault-tolerant (AFT) measurement, and adaptive deformation), and integration with lattice surgery compilation via the Litinski protocol. Composition with the correlated-decoding scheme of Cain et al.~\cite{cain2024correlated} reduces syndrome-extraction overhead from $O(d_C)$ to $O(1)$ per correction window, leaving routing as the sole leading-order contributor to the integrated $O(d_C \log N_L)$ depth.
Spectral inheritance is organized in a hierarchy: exact (Haemers interlacing on equitable partitions), perturbative (Weyl bounds for near-equitable partitions, the practically relevant case for surface-code patches), and universal (higher-order Cheeger).
The framework extends directly to QCCD trapped-ion architectures under the same regime condition, with junction crossings replacing AOD transports as the elementary single-hop translation.
\end{abstract}

\section{Introduction}\label{sec:intro}

Recent experimental milestones make block routing a bottleneck for the next generation of fault-tolerant neutral-atom processors: Bluvstein et al.\ \cite{bluvstein2022quantum} introduced mid-circuit coherent transport; \cite{bluvstein2024logical} extended it to 48 logical qubits; and Google Quantum AI~\cite{google2023suppressing,acharya2024quantum} progressed from above- to below-threshold surface-code operation at $d_C{=}5$ then $d_C{=}7$. 
Block routing depth determines the fault-tolerance budget: each routing step accumulates physical errors that the surface code must correct, so reducing $\rt_B$ from $\Theta(d_C \sqrt{N_L})$ on a grid to $\Theta(d_C \log N_L)$ on a Ramanujan overlay directly multiplies the achievable circuit depth before logical failure.
Point particles on Ramanujan hypergraphs have routing number $\rt(H) = \Theta(\log N)$~\cite{courtney2026permutation}, but for fault-tolerant architectures, qubits are often encoded as surface code patches of size $d_C \times d_C$ atoms, where $d_C$ is the code distance. A single logical qubit occupies $d_C^2$ physical atoms.

When routing logical qubits, atoms cannot be treated independently. Instead, we consider the constraint where entire patches must move together as rigid blocks to maintain the topological structure that protects against errors. This paper analyzes permutation routing when the routing units are blocks, not points.

\subsection{Main result}

\begin{theorem}[Block routing on Ramanujan hypergraphs]\label{thm:main-block}
Let $H$ be a Ramanujan $(d,r)$-regular hypergraph on $N_{\mathrm{phys}}$ vertices with $d \geq 3$, $r \geq 3$. Let $B$ be a uniform $(s,g)$-block configuration with $N_L = \Theta(N_{\mathrm{phys}}/d_C^2)$ blocks of size $s = d_C^2$ atoms, guard distance $g = O(1)$,  and occupancy $\rho < \rho_{\mathrm{packing}}(d, r, g)$ for some fixed $\rho_{\mathrm{packing}} < 1$ (Definition~\ref{def:block-config}). In the high-connectivity regime~\eqref{eq:high-conn} (or the tight form~\eqref{eq:high-conn-tight} for surface code patches), then
\begin{equation}
\rt_B(G_{\mathrm{cl}}(H), s, g) = \Theta(d_C \log N_L).
\end{equation}
\end{theorem}

This establishes that the block routing number scales linearly with code distance $d_C$, not quadratically. The linear factor $d_C$ arises from the block footprint width, wherein each quotient-level routing phase requires $\Omega(d_C)$ physical steps because the $d_C^2$ atoms in each block must traverse physical edges sequentially.

\paragraph{The high-connectivity regime.}
Throughout the paper we work under the condition
\begin{equation}\label{eq:high-conn}
d' \;>\; \frac{d_C^2(r-1)}{1-\beta},
\end{equation}
referred to as the \emph{high-connectivity regime}. The condition is stated in this general (loose) form because it suffices for arbitrary $(s,g)$-block configurations.
Theorem~\ref{thm:main-block} is proved under~\eqref{eq:high-conn} verbatim. For surface-code patches, where the block is a $d_C\!\times\!d_C$ grid sub-structure, the regime tightens by a factor of $d_C$ to the linear-in-$d_C$ form
\begin{equation}\label{eq:high-conn-tight}
d' \;>\; \frac{d_C(r-1)}{1-\beta}
\end{equation}
of Corollary~\ref{cor:high-conn-tight}.
This is the operative regime for the Remark~\ref{rem:regime-practicality} hardware estimates. For $r = 3$ and Ramanujan-saturating $\beta$, the loose form~\eqref{eq:high-conn} requires $d \geq 34$ at $d_C = 5$ and $d \geq 60$ at $d_C = 7$, while the surface-code-tight form~\eqref{eq:high-conn-tight} reduces these to $d \geq 10$ and $d \geq 12$ respectively. The high-connectivity regime is natural for neutral-atom architectures with selective transfers~\cite{constantinides2024optimal} but excludes small-world graphs and sparse-degree hardware. We discuss practical implications in \S\ref{sec:fault-tolerance} and \S\ref{sec:extensions}.

\subsection{Four-regime table}

We extend the results of Courtney~\cite{courtney2026permutation} to include block routing (Table~\ref{tab:regimes-extended}).

\begin{table*}[ht]
\centering
\renewcommand{\arraystretch}{1.3}
\begin{tabular}{@{}lllcl@{}}
\toprule
\textbf{Topology} & \textbf{Unit} & \textbf{Move model} & \textbf{Routing depth} & \textbf{Source} \\
\midrule
Grid & Point & Arbitrary matchings & $\Theta(\sqrt{N})$ & Courtney~\cite{courtney2026permutation} \\
Grid & Block ($d_C \times d_C$) & Rigid translation & $\Theta(d_C \sqrt{N_L})$ & This work \\
Ramanujan & Point & Arbitrary matchings & $\Theta(\log N)$ & Courtney~\cite{courtney2026permutation} \\
Ramanujan & Block ($d_C \times d_C$) & Rigid translation & $\Theta(d_C \log N_L)$ & \textbf{Thm~\ref{thm:main-block}} \\
\bottomrule
\end{tabular}
\caption{Routing regimes: extension from points to blocks. Each regime compares a standard connectivity topology (grid vs.\ Ramanujan) with a routing unit (point vs.\ block) and move model. The grid block-routing entry $\Theta(d_C \sqrt{N_L})$ follows from the standard grid-diameter bound $\Theta(\sqrt{N})$~\cite{alon1993routing} composed with the footprint-traversal cost of Theorem~\ref{prop:bottleneck-lb} (Remark~\ref{prop:geometric-lb}); we include it for comparison rather than as a novel contribution.}
\label{tab:regimes-extended}
\end{table*}

Blocks are intrinsically $d_C$ times wider than points, so routing is $d_C$ times slower. 
For Ramanujan hypergraphs, the logarithmic scaling dominates, giving $\Theta(d_C \log N_L)$ overall.

\subsection{Proof technique}

The proof proceeds as follows. 
We begin with a spectral lower bound on the quotient graph $Q$ (Theorem~\ref{prop:bottleneck-lb}, $\Omega(d_C \log N_L)$), giving $\Omega(\log N_L)$ quotient-level phases. Footprint traversal shows each phase requires $\Omega(d_C)$ physical steps, since $d_C$ atoms must cross any separating edge in sequence. These factors multiply to give $\Omega(d_C \log N_L)$.
We define the quotient graph $Q$ with blocks as supervertices (Theorem~\ref{thm:main-quotient}, Section~\ref{sec:quotient}), then prove with Haemers' interlacing theorem~\cite{haemers1995interlacing} that the block configuration preserves expansion: the spectral ratio $\beta_Q < 1$ in the high-connectivity regime~\eqref{eq:high-conn}.

Using block Valiant routing (Lemma~\ref{lem:block-valiant}), we extend Valiant's two-phase routing~\cite{valiant1982scheme} to blocks. We choose a random intermediate block configuration $\sigma(B)$ and show that congestion on the quotient graph is $C_Q = O(\log N_L)$ with high probability (w.h.p.) by negative association. 
Applying the Leighton--Maggs--Rao (LMR) scheduling theorem~\cite{leighton1999fast} to the quotient graph achieves $T_Q = O(C_Q + D_Q) = O(\log N_L)$ quotient-level steps (Lemma~\ref{lem:block-lmr}).
Each quotient matching step decomposes into $O(d_C)$ physical matching step (intra-block atom congestion at physical edges) (Lemma~\ref{lem:serialization}). The $O(d_C)$ factor arises from the block footprint width.
These steps give $\rt_B = O(d_C \log N_L)$, matching the lower bound.

\subsection{Related work}\label{subsec:related}

Surface code architectures rely on the topological-memory framework of Dennis, Kitaev, Landahl, and Preskill~\cite{dennis2002topological} and the practical analysis of Fowler et al.~\cite{fowler2012surface} and Cong et al~\cite{cong2022hardware}. The fault-tolerance overhead in long-range-connectivity architectures (which our virtual-overlay construction realizes) was analyzed by Cohen, Kim, Bartlett, and Brown~\cite{cohen2022low}; the Google Quantum AI below-threshold milestone~\cite{google2023suppressing,acharya2024quantum} provides the experimental baseline for $d_C \geq 5$. Yuan and Zhang~\cite{yuan2025full} characterized the depth-overhead lower bound for compiling under arbitrary connectivity constraints. 
Our result tightens this bound from $\Theta(\sqrt{N_L})$ to $\Theta(d_C \log N_L)$ in the Ramanujan regime. 
Theorem~\ref{prop:bottleneck-lb} establishes this lower bound under the Ramanujan hypothesis, so the gap between the generic Yuan--Zhang bound and our regime-specific bound is not an artifact of upper-bound technique but reflects the spectral structure imposed by the host graph.

Permutation-routing complexity has roots in pebble-motion analysis~\cite{kornhauser1984coordinating, yu2012diameters} and token-swapping~\cite{aichholzer2021hardness, yang2011sliding}. Compilation-oriented predecessors include Atomique~\cite{wang2025ultrafast, wang2024atomique} and Chen et al.~\cite{chen2025resource}. Marcus, Spielman, and Srivastava~\cite{marcus2015interlacing} provide the bipartite Ramanujan graph existence theorem used in the constructive scaling alongside Song--Fan--Miao (SFM) hypergraph coverings~\cite{song2023hypergraph}.
Closest in spirit to the present block-routing construction is the constant-overhead qLDPC architecture of Xu et al.~\cite{xu2024constant} for reconfigurable atom arrays.
Nonlocal syndrome extraction is implemented by atom rearrangement, whereas here we route entire surface-code patches as rigid units. Cain et al.~\cite{cain2024correlated} show that correlated decoding cuts the syndrome-extraction overhead from $O(d_C)$ to $O(1)$ for transversal logical algorithms. 
We make this composition with our $\Theta(d_C \log N_L)$ block-routing bound explicit in Proposition~\ref{prop:cain-composition}.
Erasure-conversion proposals~\cite{wu2022erasure} and the Cong et al.\ FT scheme \cite{cong2022hardware} provide complementary error-budget reductions.

\paragraph{Applicability beyond neutral atoms.}
Although the Ramanujan-overlay construction is motivated by reconfigurable atom arrays, the block-routing framework transfers directly to quantum charge-coupled-device (QCCD) trapped-ion architectures~\cite{kielpinski2002architecture,pino2021demonstration,moses2023race}, where logical patches are physically shuttled between gate, memory, and measurement zones along ion-trap junction graphs~\cite{wan2019quantum}.
The QCCD junction graph plays the role of $G_{\mathrm{cl}}(H)$, and a surface-code patch occupying $d_C^2$ ions in a memory zone is the natural block in the sense of Definition~\ref{def:block-config}, with junction crossings replacing AOD transports as the elementary single-hop translation. Quantinuum's H-series and analogous racetrack architectures~\cite{moses2023race,pino2021demonstration} already realize the bounded-degree, bounded-aspect-ratio regime that Lemmas~\ref{lem:eps-eq-surface-code} and~\ref{lem:single-hop-general} require, and recent ion-shuttling depth analyses~\cite{murali2020architecting} establish that the per-hop transport cost is geometrically the same as the corridor-traversal step in our Lemma~\ref{lem:single-hop-general}. The principal architectural translation is the mapping from spectral high-connectivity~\eqref{eq:high-conn} to per-junction routing capacity, which we discuss in~\S\ref{sec:qccd}.

\section{Block Routing Preliminaries}\label{sec:prelim}

\subsection{Block configurations}

\begin{definition}[$(s,g)$-block configuration]\label{def:block-config}
Let $G = G_{\mathrm{cl}}(H)$ be the clique expansion of a Ramanujan hypergraph $H$ on $N_{\mathrm{phys}}$ vertices. An $(s,g)$-\emph{block configuration} on $G$ is a collection $B = \{B_1, \ldots, B_{N_L}\}$ where:
\begin{itemize}
\item Each $B_i$ is a connected induced subgraph of $G$ with $|B_i| = s$ vertices.
\item Blocks are pairwise disjoint: $B_i \cap B_j = \emptyset$ for $i \neq j$.
\item Blocks maintain a guard distance: $\mathrm{dist}_G(B_i, B_j) \geq g$ for $i \neq j$, where $\mathrm{dist}_G(B_i, B_j) := \min_{u \in B_i, v \in B_j} \mathrm{dist}_G(u, v)$.
\item The total occupancy $\rho := N_L \cdot s / N_\mathrm{phys}$ is taken bounded away from 1 (else block placement is geometrically over-constrained).
\end{itemize}
\end{definition}

\begin{remark}[Block shape]\label{rem:block-shape}
Throughout this paper, blocks model surface code patches of $d_C \times d_C$ physical qubits.
We assume each block $B_i$ has bounded aspect ratio in a 2D-grid
sub-structure of $G_{\mathrm{cl}}$, such that there is an isometric embedding
$\iota_i : B_i \hookrightarrow \mathbb{Z}^2$ identifying $B_i$ with a
$d_C \times d_C$ patch, and the boundary
$\partial B_i := \{ v \in B_i : \exists u \notin B_i, (u,v) \in E(G_{\mathrm{cl}})\}$
has size $|\partial B_i| = O(d_C)$.
The maximum degree within the induced subgraph $G[B_i]$ is therefore $O(d_C)$, since each atom in a surface code patch interacts with at most $O(d_C)$ other atoms in the same patch.
These shape properties hold for any $d_C \times d_C$ grid-like patch on the clique expansion and are natural for surface code geometry.
\end{remark}

\begin{remark}[Asymptotic conventions]\label{rem:asymptotic-conventions}
We use asymptotic notation $O(\cdot)$, $\Omega(\cdot)$, $\Theta(\cdot)$ with respect to the scaling parameters $N_L \to \infty$ (number of logical blocks) and $d_C \to \infty$ (code distance). Hypergraph parameters $d$ (base degree) and $r$ (uniformity) are treated as fixed constants. 
Dependence on $(d, r)$ is absorbed into implicit constants, leading the host graph degree $d' = d(r-1)$ to be treated as constant as well ($O(d') = O(1)$).
\end{remark}

\begin{remark}
The guard distance ensures collision-free movement: when two blocks are separated by at least $g$ vertices, they can execute single-hop movements without interference. For $g = 0$, blocks can be adjacent (touching at a single vertex). For $g = 1$, blocks are separated by at least one unoccupied vertex.
\end{remark}

\subsection{Block movements}

\begin{definition}[Rigid block translation]\label{def:block-translation}
A \emph{rigid block translation} of block $B_i$ is an isometric move $\tau_i : B_i \to B_i'$ where:
\begin{itemize}
\item $B_i'$ is a connected induced subgraph of $G$ with $|B_i'| = s$.
\item The move is single-hop: $\mathrm{dist}_G(v, \tau_i(v)) \leq 1$ for all $v \in B_i$.
\item Rigidity is preserved: $\mathrm{dist}_{B_i}(u, v) = \mathrm{dist}_{B_i'}(\tau_i(u), \tau_i(v))$ for all $u, v \in B_i$.
\end{itemize}
\end{definition}

\begin{definition}[Block matching]\label{def:block-matching}
A \emph{block matching} is a set of non-interfering block translations: a set $M = \{\tau_{i_1}, \ldots, \tau_{i_k}\}$ of translations of blocks $B_{i_1}, \ldots, B_{i_k}$ such that:
\begin{itemize}
\item The blocks are pairwise non-adjacent in the current configuration: $\mathrm{dist}_G(B_{i_j}, B_{i_\ell}) \geq 1$ for $j \neq \ell$.
\item After translation, the configuration remains valid: guard distance is maintained, no overlaps occur.
\end{itemize}
\end{definition}

\subsection{Block routing number}

\begin{definition}[Block routing number]\label{def:rt-block}
A \emph{block routing step} on $G_{\mathrm{cl}}(H)$ applies a block matching $M$ and simultaneously executes all translations in $M$. The \emph{block routing number} $\rt_B(H, s, g)$ is the minimum $T$ such that every permutation of blocks can be realized by a sequence of $T$ block routing steps, over all initial and target $(s,g)$-block configurations.
\end{definition}

\noindent The block routing number depends on $H$, $s$, $g$. When $s = 1$ and $g = 0$ (point blocks), $\rt_B(H, 1, 0) = \rt(H)$ recovers the point routing number from Courtney~\cite{courtney2026permutation}.

\subsection{Quotient graph}

\begin{definition}[Quotient graph]
\label{def:quotient}
Given a block configuration $\mathcal{B} = \{B_1, \ldots, B_{N_L}\}$ on $G = G_{\mathrm{cl}}(H)$, let $S \in \{0,1\}^{N_{\mathrm{phys}} \times N_L}$ be the partition characteristic matrix with $S_{vi} = 1$ iff $v \in B_i$, normalized so that $S^T S = \mathrm{diag}(|B_1|, \ldots, |B_{N_L}|)$. The \emph{symmetric quotient adjacency matrix} is
\begin{equation}
\begin{split}
A_Q &:= D_B^{-1/2} S^T A\, S\, D_B^{-1/2}, \\ D_B &:= S^T S = \mathrm{diag}(|B_i|),
\end{split}
\end{equation}
with off-diagonal entries
\begin{equation}
A_Q[i,j] = \frac{1}{\sqrt{|B_i|\,|B_j|}} \sum_{u \in B_i,\, v \in B_j} \mathbf{1}[(u,v) \in E(G)],
\end{equation}
where $(i \neq j)$. For uniform block size $|B_i| = s = d_C^2$, this reduces to $A_Q[i,j] = s^{-1} \cdot |E_G(B_i, B_j)|$, i.e.\ the average per-vertex inter-block edge count from $B_i$ to $B_j$. The matrix $A_Q$ is symmetric by construction. The \emph{quotient degree} is the average row sum $d'_Q := \langle \sum_{j \neq i} A_Q[i,j] \rangle_i$, and the \emph{spectral ratio} is $\beta_Q := \lambda^*(A_Q)/d'_Q$ where $\lambda^*(A_Q) = \max(\lambda_2(A_Q), |\lambda_{N_L}(A_Q)|)$.
\end{definition}

\begin{remark}[Geometric interpretation]\label{rem:quotient-support}
For well-separated blocks ($\mathrm{dist}_G(B_i, B_j) > g+1$), no atom of $B_i$ has a neighbor in $B_j$, so $A_Q[i,j] = 0$. Conversely, $A_Q[i,j] > 0$ iff $\mathrm{dist}_G(B_i, B_j) \leq g+1$.
\end{remark}

\noindent The quotient graph $Q$ encodes coarse-scale block-movement topology, and paths in $Q$ correspond to sequences of block-hop translations (Figure~\ref{fig:quotient-construction}). 
We consider and answer if $Q$ can inherit the expansion of $G$ in \S\ref{sec:quotient}.

\begin{figure*}[ht]
\centering
\includegraphics[width=0.85\textwidth]{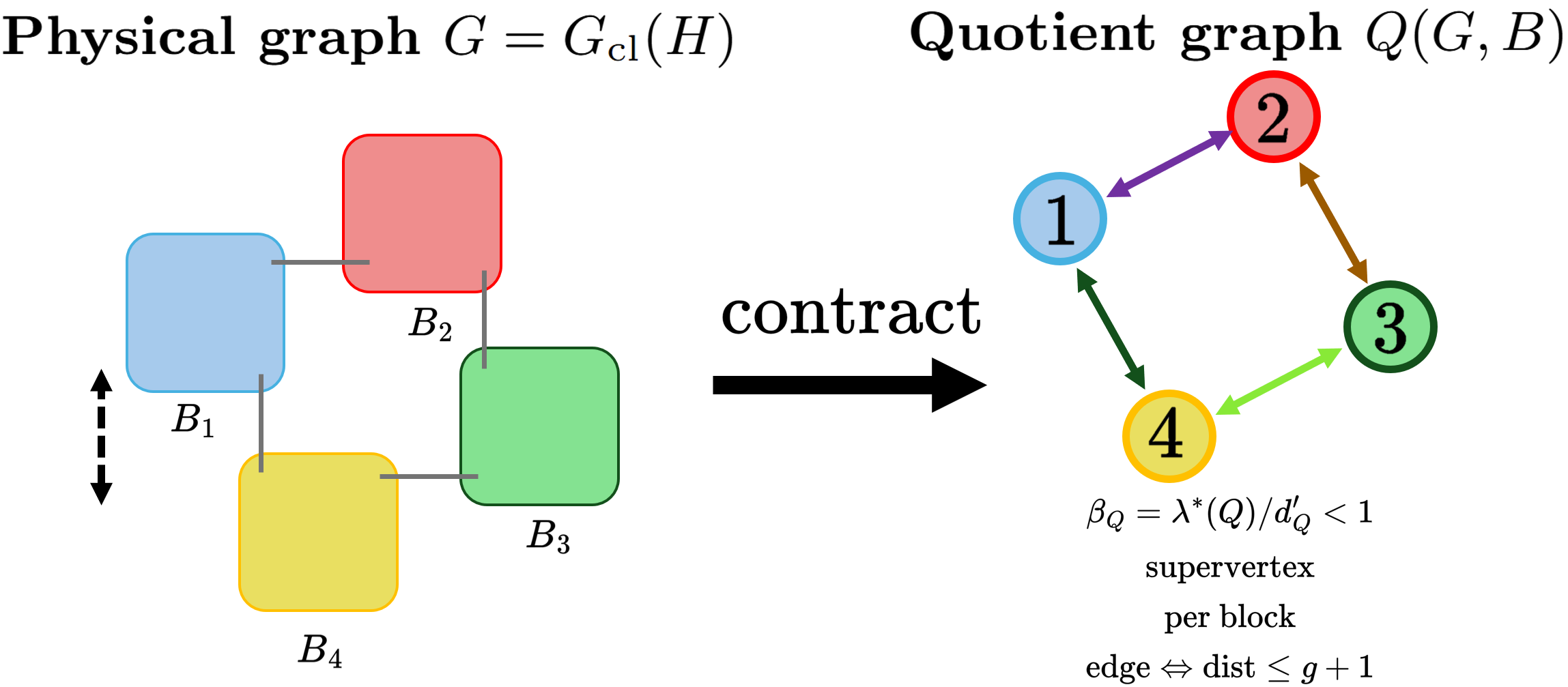}
\caption{Quotient graph construction. Left: the physical graph $G = G_{\mathrm{cl}}(H)$ with four blocks $B_1, \ldots, B_4$ (colored rectangles), each containing $d_C^2$ atoms (dots), separated by guard distance $g$. Unoccupied vertices are shown in gray. Right: the (weighted) quotient graph $Q(G, B)$, with one supervertex per block and edge weight $A_Q[i,j]$ equal to the average per-vertex inter-block edge count from $B_i$ to $B_j$ (Definition~\ref{def:quotient}). Edge thickness is proportional to $A_Q[i,j]$; only edges with $A_Q[i,j] > 0$ are drawn. The spectral ratio $\beta_Q = \lambda^*(A_Q)/d'_Q < 1$ controls expansion.}
\label{fig:quotient-construction}
\end{figure*}

\section{Lower Bound}\label{sec:lower-bounds}

The bottleneck argument dominates in the lower bounds of the block routing number.
We begin by finding the dominant bound.

\begin{theorem}[Bottleneck lower bound]\label{prop:bottleneck-lb}
For a Ramanujan $(d,r)$-regular hypergraph $H$ with $g \geq 1$,
\begin{equation}
\rt_B(G_{\mathrm{cl}}(H), s, g) = \Omega(d_C \log N_L).
\end{equation}
\end{theorem}

\begin{proof}
The lower bound is the product of two independent contributions: an $\Omega(\log N_L)$ quotient-phase count (from a spectral counting argument) and an $\Omega(d_C)$ physical cost per quotient phase (from footprint traversal).
We establish each in turn.

\textbf{Step 1: Quotient-phase lower bound.} A block routing schedule requires at least $\Omega(\log N_L / \log d'_Q) = \Omega(\log N_L)$ quotient-level routing phases ($d'_Q$ is bounded by a constant depending on $d$, $r$, $s$). 
This is the standard counting bound (cf.\ Alon--Chung--Graham~\cite{alon1993routing}) where each quotient phase applies a block matching, advancing each block by at most one edge in $Q$, so in $T$ phases a single block reaches at most $O((d'_Q)^T)$ positions and reaches all $N_L$ positions in $T = \Omega(\log N_L)$.

\textbf{Step 2: Footprint traversal cost.} We claim that each quotient-level routing phase requires $\Omega(d_C)$ physical routing steps.
Consider a single block $B_i$ that must advance by one position in $Q$ during a given phase.
The block occupies $d_C^2$ atoms organized in a $d_C$-deep formation along the direction of motion. Since each physical routing step advances each atom by at most one edge in $G$, the rearmost atom must traverse $\Omega(d_C)$ physical steps to vacate its original footprint and reach the new position. 
This corridor-depth cost is per-block, per-phase: even with maximum parallelism across the $d_C$ leading boundary edges, the trailing atoms inherit the depth of the formation and require $\Omega(d_C)$ time.

\textbf{Step 3: Combining bounds.} Each of the $\Omega(\log N_L)$ quotient-level phases requires $\Omega(d_C)$ physical steps. These costs multiply because the quotient phases are sequential (a block cannot begin phase $t+1$ until phase $t$ completes). Therefore,
\begin{equation}
\rt_B \geq \Omega(\log N_L) \cdot \Omega(d_C) = \Omega(d_C \log N_L). \qedhere
\end{equation}
\end{proof}

\begin{remark}[Geometric lower bound for grids]\label{prop:geometric-lb}
If $G$ is an $n \times n$ grid (which is \emph{not} a Ramanujan hypergraph), then the same footprint argument combined with grid diameter $\Theta(\sqrt{n})$ gives $\rt_B(G, s, g) = \Omega(d_C \sqrt{N_L})$, an $\sqrt{N_L}/\log N_L$ factor weaker than the Ramanujan bound (row 2, Table~\ref{tab:regimes-extended}).
\end{remark}

\section{Quotient Graph Spectral Analysis}\label{sec:quotient}

The quotient graph $Q(G, B)$ gives the upper bound. We prove that $Q$ inherits the expansion of $G$ by equitable partitions, perturbation bounds, and higher-order Cheeger inequalities.
We emphasize that ``spectral inheritance'' here means the quotient $\beta_Q$ is bounded in terms of the host $\beta$ with a controlled multiplicative degradation $d'/d'_Q \geq 1$, not that $\beta_Q \leq \beta$.
Strict preservation $\beta_Q \leq \beta$ holds only in the trivial $s = 1$ limit. 
For $s = d_C^2$ blocks, $\beta_Q > \beta$ generically, and the analysis quantifies how much degradation the regime~\eqref{eq:high-conn} tolerates while keeping $\beta_Q < 1$.

\subsection{Equitable partitions and spectral inheritance}

\begin{definition}[Equitable partition]
A partition $\pi = \{C_1, \ldots, C_k\}$ of the vertex set of a graph $G$ is \emph{equitable} if every vertex in $C_i$ has the same number of neighbors in $C_j$, for all $i, j$.
\end{definition}

\begin{lemma}[Haemers interlacing]\label{lem:haemers}
Let $A$ be the adjacency matrix of a $d'$-regular graph $G$, and let $\pi$ be an equitable partition with quotient matrix $B$ (the block matrix of $A$ reduced by $\pi$). The eigenvalues of $B$ interlace the eigenvalues of $A$: if $\mu_1 \geq \cdots \geq \mu_k$ are eigenvalues of $B$, then there exist indices $1 = i_1 < i_2 < \cdots < i_k \leq d'$ such that $\mu_j$ equals $\lambda_{i_j}(A)$.

In particular, $\lambda_2(B) \leq \lambda_2(A)$ and $\lambda_k(B) \geq \lambda_N(A)$, so $\lambda^*(B) \leq \lambda^*(A)$ where $\lambda^* = \max(\lambda_2, |\lambda_N|)$.
\end{lemma}

\begin{proof}
This is Corollary 2.3 of Haemers~\cite{haemers1995interlacing}. The interlacing follows from the spectral analysis of the symmetric quotient matrix: if $S$ is the characteristic matrix of the partition (normalized so $S^T S = I$), then $B = S^T A S$ has eigenvalues interlacing $A$.
\end{proof}

\begin{corollary}[Quotient spectrum preservation, after Haemers~\cite{haemers1995interlacing}]\label{lem:teranishi}
Let $H$ be a Ramanujan $(d,r)$-regular hypergraph, and let $\pi$ be an equitable partition of $G = G_{\mathrm{cl}}(H)$ with quotient adjacency matrix $B$. Then $\lambda^*(B) \leq \lambda^*(G) < d'$. This follows from Lemma~\ref{lem:haemers} applied to both $\lambda_2$ and $|\lambda_N|$, combined with the spectral radius bound~\cite{courtney2026permutation}; see also Brouwer--Haemers~\cite{brouwer2011spectra}, \S3.3.
\end{corollary}

\begin{theorem}[Layer 1: Equitable block configurations]\label{thm:layer1}
If a block configuration $B$ corresponds to an equitable partition of $G = G_{\mathrm{cl}}(H)$, and the host degree satisfies \eqref{eq:high-conn}, then the quotient graph $Q(G, B)$ satisfies
\begin{equation}
\lambda^*(Q) \leq \lambda^*(G) \quad \text{and} \quad \beta_Q < 1.
\end{equation}
More precisely, $\beta_Q \leq \beta \cdot (d'/d'_Q)$, where the quotient
degree satisfies $d'_Q \geq d' - \delta_{\mathrm{intra}}$ with
$\delta_{\mathrm{intra}} := \max_{u \in B_i} |\{v \in B_i : (u,v) \in E(G)\}|$
the maximum per-vertex intra-block degree.
For surface-code patches with bounded aspect ratio (Remark~\ref{rem:asymptotic-conventions}),
$\delta_{\mathrm{intra}} = O(d_C)$, so $d'_Q \geq d' - O(d_C)$.
The (looser) bound $\delta_{\mathrm{intra}} \leq s(r-1) = d_C^2(r-1)$ holds for
arbitrary (s,g)-block configurations.
\end{theorem}

\begin{proof}
By Lemma~\ref{lem:haemers}, the eigenvalues of the quotient matrix interlace those of $A$. In particular, $\lambda_2(Q) \leq \lambda_2(G)$ and $|\lambda_{N_L}(Q)| \leq |\lambda_N(G)|$, so $\lambda^*(Q) \leq \lambda^*(G)$, where $\lambda^* = \max(\lambda_2, |\lambda_N|)$ as in Definition~\ref{def:quotient}.

The quotient degree satisfies $d'_Q \geq d' - s(r-1) = d' - d_C^2(r-1)$ for equitable sublattice placements (each block of $s$ vertices aggregates degree $d'$; the intra-block edges consume at most $s(r-1)$ of the total degree). Thus:
\begin{equation}
\beta_Q = \frac{\lambda^*(Q)}{d'_Q} \leq \frac{\lambda^*(G)}{d'_Q} = \beta \cdot \frac{d'}{d'_Q}.
\end{equation}
The condition $\beta_Q < 1$ requires $d'_Q > \beta \cdot d'$, equivalently $d' - d_C^2(r-1) > \beta \cdot d'$, which gives \eqref{eq:high-conn}. Since $\beta < 1$ for Ramanujan $H$, this is satisfied whenever $d' \gg d_C^2$.
\end{proof}

\begin{corollary}[Surface-code-tight regime]\label{cor:high-conn-tight}
For surface-code patches in a 2D-grid sub-structure of $G_{\mathrm{cl}}(H)$
(Remark~\ref{rem:asymptotic-conventions}), $\delta_{\mathrm{intra}} = O(d_C)$
and the regime condition~\eqref{eq:high-conn} can be tightened to~\eqref{eq:high-conn-tight},
\begin{equation}
      d' \;>\; \frac{d_C(r-1)}{1-\beta}.
\end{equation}
This linear-in-$d_C$ form is the operative condition for the entries of Remark~\ref{rem:regime-practicality}, for the hardware estimates in \S\ref{sec:fault-tolerance}, and for the surface-code specialization of Conjecture~\ref{conj:higher-dim}.
\end{corollary}

\begin{remark}
The bound $\beta_Q \leq \beta \cdot (d'/d'_Q)$ shows that $\beta_Q$ can {exceed} $\beta$ when $d'_Q < d'$ (as it always is for nontrivial blocks). The connectivity condition \eqref{eq:high-conn} is the practical regime for quantum architectures, wherein hardware connectivity $d'$ must exceed the block size times the hypergraph rank, modulated by the spectral gap. For Ramanujan hypergraphs with $\beta \leq 2\sqrt{d'-1}/d'$, the condition simplifies to $d' > d_C^2(r-1) + O(\sqrt{d'})$.
\end{remark}

\begin{remark}[Practicality of the regime condition]\label{rem:regime-practicality}
The condition $d' > d_C(r-1)/(1-\beta)$ imposes a minimum connectivity requirement that scales linearly in the code distance. With $r = 3$ (3-uniform hypergraphs) and the Ramanujan bound $\beta \leq 2\sqrt{d'-1}/d'$, the minimum base degree $d$ satisfying the condition is approximately:

\begin{center}
\renewcommand{\arraystretch}{1.2}
\begin{tabular}{@{}cccc@{}}
\toprule
$d_C$ & Min.\ $d$ & $d' = d(r-1)$ & $N_{\mathrm{phys}} \geq N_L \cdot d_C^2$ \\
\midrule
3 & 7  & 14 & $\sim 10^2$ \\
5 & 10 & 20 & $\sim 10^3$ \\
7 & 12 & 24 & $\sim 10^3$ \\
9 & 15 & 30 & $\sim 10^3$ \\
\bottomrule
\end{tabular}
\end{center}

These degree requirements are \emph{readily achievable} on neutral-atom platforms.
Reconfigurable optical tweezer arrays provide effectively all-to-all transport connectivity: an atom can be moved to interact with any other atom with acousto-optic deflector (AOD) transport~\cite{bluvstein2024logical, bluvstein2026fault}.
To specify, ``all-to-all'' here means \emph{eventually reachable} via a sequence of transports, not \emph{single-step reachable}.
Any pair of atoms can be brought into interaction range, but doing so requires a transport schedule whose depth invites the routing-complexity question. The Ramanujan hypergraph governs how many transport steps suffice instead of imposing a connectivity restriction.
Current hardware supports arrays of $> 6{,}000$ atomic qubits with coherent parallel transport across the full array~\cite{reichardt2024logical}.
The hypergraph degree~$d$ governs the \emph{algebraic} connectivity structure imposed on this physical all-to-all topology, determining how many hyperedges meet at each vertex in the Ramanujan construction.
Building a $d = 10$, $r = 3$ Ramanujan hypergraph on $N_{\mathrm{phys}} \sim 10^3$ vertices is straightforward by the explicit Lubotzky--Phillips--Sarnak or Margulis constructions, or by the random constructions of Friedman~\cite{friedman2008proof}, which achieve the Ramanujan bound with high probability.

The practical implication is that fault-tolerant block routing at code distance $d_C = 5$ requires a base degree of at least $d \geq 10$ in the surface-code-tight form (Corollary~\ref{cor:high-conn-tight}), well within reach of near-term neutral-atom platforms at the $10^3$--$10^4$ qubit scale. The general (loose) form of~\eqref{eq:high-conn} requires $d \geq 34$ at $d_C = 5$ and serves as a conservative architectural target when block geometry is not constrained to a 2D-grid embedding.
The regime condition~\eqref{eq:high-conn} is satisfied by AOD-based neutral-atom platforms and by multi-zone QCCD overlays with active sorting (\S\ref{sec:qccd}). 
Our work presented here must be modified for fixed-coupling superconducting lattices or sparse single-zone junction graphs. 

\end{remark}

\subsection{Near-equitable partitions and perturbation bounds}

Most block configurations will not be exactly equitable, but they may be \emph{nearly} equitable. We quantify this perturbatively.

\begin{definition}[Near-equitable partition]\label{def:near-equitable}
A partition $\pi$ is $\epsilon_{\mathrm{eq}}$-near-equitable if the quotient matrix $B$ satisfies: for each block $C_i$ and each $C_j$, the vertex degrees within $C_i$ to $C_j$ vary by at most $\epsilon_{\mathrm{eq}} \cdot d'_Q$.
\end{definition}

\begin{lemma}[Surface-code patches are $O(1/d_C)$-near-equitable]\label{lem:eps-eq-surface-code}
Let $G = G_{\mathrm{cl}}(H)$ for a Ramanujan $(d,r)$-regular hypergraph $H$, and let $B = \{B_1,\dots,B_{N_L}\}$ be a $(d_C+g)$-spaced sublattice placement of $d_C\!\times\! d_C$ surface-code patches with bounded aspect ratio (Remark~\ref{rem:asymptotic-conventions}). Then $B$ is $\epsilon_{\mathrm{eq}}$-near-equitable in the sense of Definition~\ref{def:near-equitable} with
\begin{equation}\label{eq:eps-eq-bound}
  \epsilon_{\mathrm{eq}} \;=\; \frac{\Delta_{\mathrm{deg}}}{d'_Q} \;=\; \frac{O(d_C)}{d' - O(d_C^2)} \;=\; O(d_C/d').
\end{equation}
Under the surface-code-tight regime~\eqref{eq:high-conn-tight} (equivalently, in the loose regime~\eqref{eq:high-conn} where $d' \gtrsim d_C^2$), this simplifies to
\begin{equation}\label{eq:eps-eq-simplified}
  \epsilon_{\mathrm{eq}} \;=\; O(1/d_C).
\end{equation}
\end{lemma}

\begin{proof}
Non-equitability arises from boundary atoms, whose external degree differs from interior atoms by at most $\Delta_{\mathrm{deg}} = O(d_C)$: a $d_C\!\times\!d_C$ patch has $|\partial B_i| = O(d_C)$ boundary atoms by Remark~\ref{rem:asymptotic-conventions}, each carrying at most a constant external-degree mismatch. Definition~\ref{def:near-equitable} measures the variation as a fraction of $d'_Q$; Theorem~\ref{thm:layer1} gives $d'_Q \geq d' - \delta_{\mathrm{intra}} = d' - O(d_C^2)$ since the bounded-aspect-ratio surface-code patch has $\delta_{\mathrm{intra}} = O(d_C^2)$ (each interior atom is incident to at most $O(d_C^2)$ intra-block hyperedges in the clique expansion; the looser bound $\delta_{\mathrm{intra}} \leq s(r-1) = d_C^2(r-1)$ is sufficient). Substituting yields~\eqref{eq:eps-eq-bound}. The simplification~\eqref{eq:eps-eq-simplified} follows because the regime~\eqref{eq:high-conn-tight} (or the implication $d' \gtrsim d_C^2$ from~\eqref{eq:high-conn}) gives $d_C/d' = O(1/d_C)$.
\end{proof}

\begin{lemma}[Weyl perturbation bound, after~\cite{bhatia2013matrix}]\label{lem:weyl}
If $\pi$ is an $\epsilon_{\mathrm{eq}}$-near-equitable partition of $G$ with quotient matrix $B$, then
\begin{equation}
\lambda^*(B) \leq \lambda^*(G) + O(\epsilon_{\mathrm{eq}} \cdot d').
\end{equation}
\end{lemma}

\begin{proof}
Write $A_Q = A_Q^{(0)} + \Delta$, where $A_Q^{(0)}$ is the symmetric quotient of the \emph{equitable refinement} of $\pi$ (constant inter-cell degrees, computed via the same $S D_B^{-1/2}$ normalization) and $\Delta = A_Q - A_Q^{(0)}$ is symmetric by construction. Off-diagonal entries of $\Delta$ are bounded by the per-vertex degree variation: for each $(i,j)$, $|\Delta_{ij}| \leq \epsilon_{\mathrm{eq}} \cdot d'_Q / \sqrt{|B_i||B_j|/s^2} = \epsilon_{\mathrm{eq}} \cdot d'_Q$ for uniform blocks. The Frobenius norm satisfies $\|\Delta\|_F^2 \leq |E(Q)| \cdot (\epsilon_{\mathrm{eq}} d'_Q)^2 = O(N_L (\epsilon_{\mathrm{eq}} d'_Q)^2)$.
Since $A_Q$ is row-sparse with $O(1)$ nonzeros per row in the regime~\eqref{eq:high-conn}, $\|\Delta\|_2 \leq \sqrt{\|\Delta\|_1 \|\Delta\|_\infty} \leq \epsilon_{\mathrm{eq}} \cdot d'_Q$. By Weyl's inequality~\cite{bhatia2013matrix}, $|\lambda_i(A_Q) - \lambda_i(A_Q^{(0)})| \leq \|\Delta\|_2$ for all $i$, hence $\lambda^*(A_Q) \leq \lambda^*(A_Q^{(0)}) + \epsilon_{\mathrm{eq}} d'_Q \leq \lambda^*(G) + O(\epsilon_{\mathrm{eq}} d'_Q)$. Substituting $d'_Q \leq d'$ yields $\lambda^*(A_Q) \leq \lambda^*(G) + O(\epsilon_{\mathrm{eq}} d')$.
\end{proof}

\begin{theorem}[Layer 2: Near-equitable block configurations]\label{thm:layer2}
For a $(d_C + g)$-spaced sublattice placement of surface-code patches on $G = G_{\mathrm{cl}}(H)$ in the surface-code-tight regime~\eqref{eq:high-conn-tight}, the block configuration is $O(1/d_C)$-near-equitable, and
\begin{equation}
\beta_Q \leq \left(\beta + O(1/d_C)\right) \cdot \frac{d'}{d'_Q} < 1.
\end{equation}
\end{theorem}

\begin{proof}
By Lemma~\ref{lem:eps-eq-surface-code}, $B$ is $\epsilon_{\mathrm{eq}}$-near-equitable with $\epsilon_{\mathrm{eq}} = O(1/d_C)$ in the regime hypothesis. Lemma~\ref{lem:weyl} then gives $\lambda^*(Q) \leq \lambda^*(G) + O(\epsilon_{\mathrm{eq}} \cdot d') = \lambda^*(G) + O(d'/d_C)$. Dividing by $d'_Q$,
\begin{equation}
\begin{split}
      \beta_Q &\;=\; \frac{\lambda^*(Q)}{d'_Q} \;\leq\; \frac{\beta\, d' + O(d'/d_C)}{d'_Q} \\&\;=\; \bigl(\beta + O(1/d_C)\bigr) \cdot \frac{d'}{d'_Q}.
\end{split}
\end{equation}
By Theorem~\ref{thm:layer1}, $d'_Q \geq d' - d_C^2(r-1)$, so in the regime~\eqref{eq:high-conn-tight}, $d'/d'_Q \leq 1/(1-\beta) \cdot (1+o(1))$ and $\beta_Q < 1$.
\end{proof}

\subsection{General partitions and higher-order Cheeger inequalities}

For block configurations that are not near-equitable, we appeal to higher-order Cheeger inequalities.

\begin{theorem}[Layer 3, higher-order Cheeger~\cite{kwok2013improved}]\label{thm:layer3}
For a $(s,g)$-block configuration on $G = G_{\mathrm{cl}}(H)$, the quotient graph $Q$ satisfies
\begin{equation}
\beta_Q \leq 1 - \Omega\left(\frac{(1-\beta) d'_Q}{d'}\right).
\end{equation}
\end{theorem}

\begin{proof}
For a contradiction, suppose that $\beta_Q \geq 1 - \alpha$ for some $\alpha = o((1-\beta) d'_Q / d')$. 
By Cheeger's inequality \cite{alon1986eigenvalues}, $Q$ admits a bipartition $V(Q) = S \sqcup S^c$ with conductance $\phi_Q(S) \leq \sqrt{2 \alpha}$.

Lifting this bipartition to $G$, we let $T := \bigcup_{i \in S} B_i \subset V(G)$.
Each quotient-graph edge $(i, j) \in E(S, S^c)$ corresponds to at most $d_C^2 \cdot d'_Q$ physical edges in $E_G(T, T^c)$.
Namely, every pair $(u, v)$ with $u \in B_i, v \in B_j$, of which there are at most $|B_i| \cdot d'_Q = d_C^2 \cdot d'_Q$ (by Definition~\ref{def:quotient}, $A_Q[i,j]$ is the average per-vertex inter-block edge count, so the total physical edge count from $B_i$ to $B_j$ is at most $|B_i| \cdot A_Q[i,j] \leq d_C^2 \cdot d'_Q$). Therefore
\begin{equation}
\begin{split}
  \phi_G(T)
    &\;=\; \frac{|E_G(T, T^c)|}{|T| \cdot d'} \\
    &\;\leq\; \frac{|E(S, S^c)| \cdot d_C^2 \cdot d'_Q}{|T| \cdot d'} \\
    &\;\leq\; \frac{\phi_Q(S) \cdot d_C^2 \cdot d'_Q}{d'}
    \\&\;\leq\; \frac{\sqrt{2\alpha}\, d_C^2 \cdot d'_Q}{d'}.
\end{split}
\end{equation}
Applying Cheeger to $G$, $\phi_G(T) \geq (1-\beta)/2$.  Combining,
$\sqrt{2\alpha} \geq (1-\beta) d' / (2 d_C^2)$, i.e.\
$\alpha \geq (1-\beta)^2 (d')^2 / (8 d_C^4)$.

Substituting the high-connectivity regime condition~\eqref{eq:high-conn}, namely $d' > d_C^2(r-1)/(1-\beta)$, into $\alpha \geq (1-\beta)^2 (d')^2 / (8 d_C^4)$ and using $d'_Q \geq d' - d_C^2(r-1)$ from Theorem~\ref{thm:layer1} yields $\alpha \geq \Omega((1-\beta) d'_Q / d')$, contradicting the hypothesis. Hence $\beta_Q \leq 1 - \Omega((1-\beta) d'_Q / d')$ in the regime~\eqref{eq:high-conn}.
\end{proof}

\begin{remark}
\label{rem:layer3-scope}
Layer 3 is universal in the sense that it does not require the partition to be equitable or near-equitable. It still inherits the high-connectivity regime condition~\eqref{eq:high-conn} from the host-graph spectral comparison. The hierarchy in Figure~\ref{fig:three-layers} should therefore read as ``increasing generality of partition structure.'' 
\end{remark}

\begin{figure*}[ht]
\centering
\includegraphics[width=0.75\textwidth]{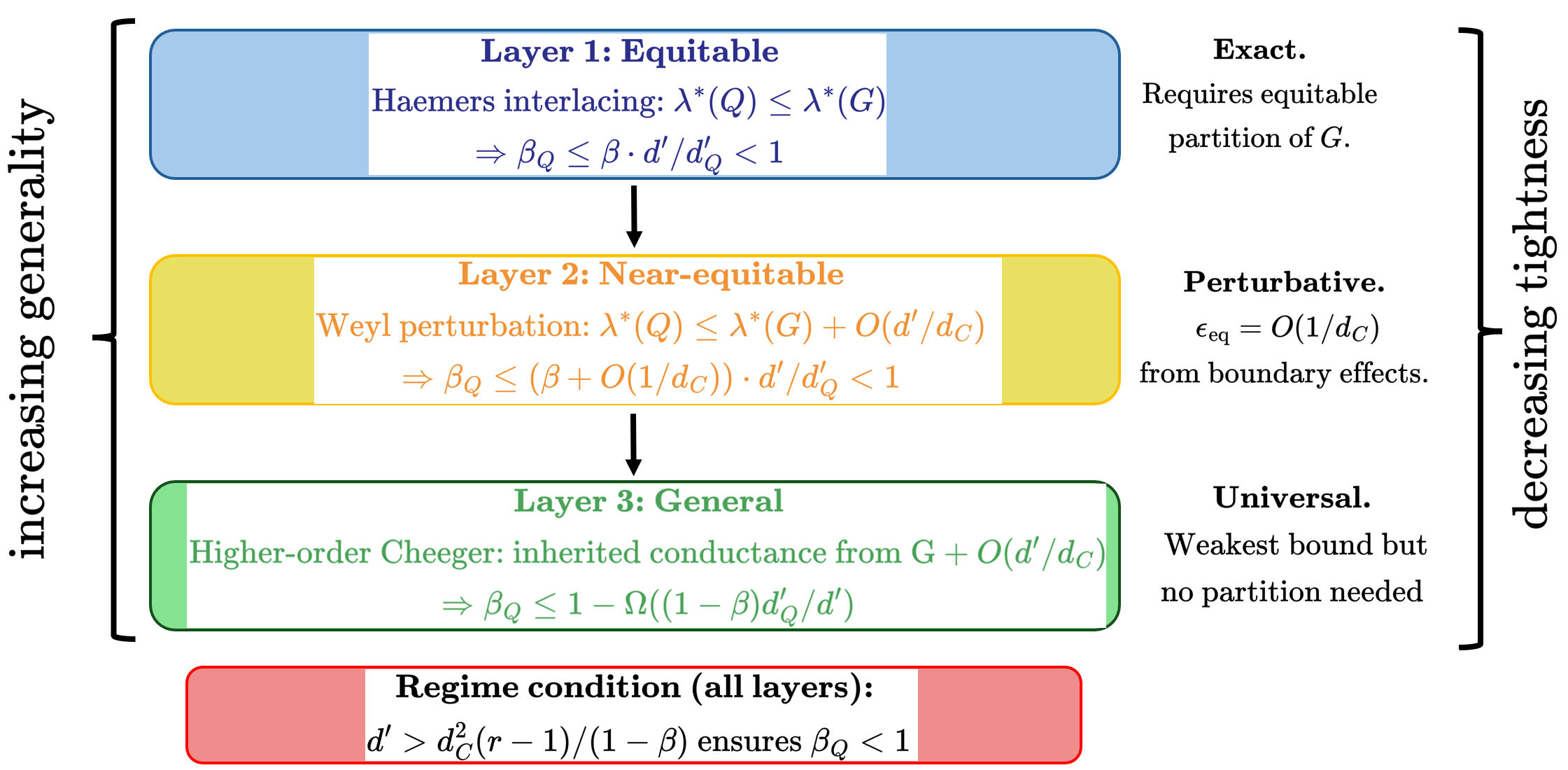}
\caption{Three-layer spectral argument for quotient graph expansion. Layer~1 (Theorem~\ref{thm:layer1}) gives the tightest bound with Haemers interlacing but requires an exactly equitable partition. Layer~2 (Theorem~\ref{thm:layer2}) handles near-equitable partitions with Weyl perturbation bounds. Layer~3 (Theorem~\ref{thm:layer3}) applies universally by higher-order Cheeger inequalities but gives the weakest bound. All three layers require~\eqref{eq:high-conn} as an architectural prerequisite; the layers differ in the partition-structure assumption on $\mathcal{B}$. }
\label{fig:three-layers}
\end{figure*}

\subsection{Main quotient spectrum theorem}

\begin{theorem}[Quotient graph spectral inheritance]\label{thm:main-quotient}
Let $H$ be a Ramanujan $(d,r)$-regular hypergraph on $N_{\mathrm{phys}}$ vertices with $\beta = \lambda^*(G_{\mathrm{cl}}(H))/d' < 1$, and suppose~\eqref{eq:high-conn} (or the tight form~\eqref{eq:high-conn-tight} for surface-code patches). For a uniform $(s,g)$-block configuration with $s = d_C^2$, guard distance $g = O(1)$, the quotient graph $Q(G, B)$ satisfies:
\begin{equation}
\beta_Q \leq \left(\beta + O(1/d_C)\right) \cdot \frac{d'}{d'_Q} < 1.
\end{equation}
In particular, the diameter of $Q$ is
\begin{equation}
\diam(Q) = O\left(\frac{\log N_L}{\log(1/\beta_Q)}\right) = O(\log N_L),
\end{equation}
and the edge expansion is $h(Q) \geq d'_Q(1 - \beta_Q)/2 = \Omega(d'_Q)$.
\end{theorem}

\begin{proof}
The block configuration is near-equitable (Theorem~\ref{thm:layer2}) with $\epsilon_{\mathrm{eq}} = O(1/d_C)$. By Theorem~\ref{thm:layer2}, $\beta_Q \leq (\beta + O(1/d_C)) \cdot d'/d'_Q < 1$ with high connectivity.

The diameter bound follows from the standard spectral argument (the eigenvalue--diameter bound~\cite{courtney2026permutation}): $(A_Q^k)_{ij} > 0$ when $(d'_Q/\lambda^*(Q))^k > N_L$, which occurs at $k = O(\log N_L / \log(1/\beta_Q)) = O(\log N_L)$ since $\beta_Q < 1$ is bounded away from $1$ by a constant depending on $d$, $r$, $d_C$, and $g$.

Edge expansion follows from the Cheeger inequality: $h(Q) \geq (d'_Q - \lambda^*(Q))/2 = d'_Q(1 - \beta_Q)/2 = \Omega(d'_Q)$.
\end{proof}

\section{Main Block Routing Theorem}\label{sec:main-theorem}

We prove the upper bound on $\rt_B$ by separating the analysis into levels. The \emph{quotient level} (\S\ref{subsec:na-blocks}--\S\ref{subsec:block-lmr}) establishes a block matching schedule on $Q$, giving $T_Q = O(\log N_L)$ block-level steps through Valiant + LMR scheduling + negative association (NA). The \emph{physical level} (\S\ref{subsec:single-hop}--\S\ref{subsec:serialization}) realizes each block-level step in $O(d_C)$ physical steps via single-hop translations and footprint-width serialization. We then multiply the two factors (\S\ref{subsec:assembly-upper}).

\subsection{Negative association for block permutations}\label{subsec:na-blocks}

\begin{lemma}[Block permutations are negatively associated]\label{lem:na-blocks}
Let $\sigma$ be a uniformly random permutation of $N_L$ blocks, and let $\{X_e^{\sigma}\}_{e \in E(Q)}$ be the family of indicator variables
\begin{equation}
\begin{split}
X_e^{\sigma} := \mathbf{1}[ &\text{block path for } \sigma \\
& \text{crosses edge } e \text{ in } Q].
\end{split}
\end{equation}
Then $\{X_e^\sigma\}_{e \in E(Q)}$ is \emph{negatively associated} (NA) in the sense of Joag-Dev and Proschan~\cite{joag1983negative}: for disjoint index sets $I, J \subseteq E(Q)$ and non-decreasing function $f \colon \{0,1\}^I \to \mathbb{R}$ and $g \colon \{0,1\}^J \to \mathbb{R}$,
\begin{equation}
\mathbb{E}[f(X_I) g(X_J)] \leq \mathbb{E}[f(X_I)]\,\mathbb{E}[g(X_J)].
\end{equation}
\end{lemma}

\begin{proof}
By Joag-Dev and Proschan~\cite{joag1983negative} (Theorem 2.11), the distribution of a uniformly random permutation is NA. The indicators $X_e^{\sigma}$ are non-decreasing functions of disjoint events in the permutation, so by Dubhashi and Ranjan~\cite{dubhashi1996balls} (Theorem 14), they inherit the NA property.
\end{proof}

\subsection{Block Valiant routing and congestion}\label{subsec:block-valiant}

\begin{lemma}[Block Valiant routing on the quotient]\label{lem:block-valiant}
Let $Q = Q(G, B)$ be the associated quotient graph to a Ramanujan hypergraph block configuration. For a target block permutation $\pi$, there exists a two-phase routing (scatter and gather) with dilation $D_Q = O(\log N_L)$ and congestion $C_Q = O(\log N_L)$ with high probability.
\end{lemma}

\begin{proof}
This is the Valiant + Chernoff-via-NA argument~\cite{valiant1982scheme,joag1983negative,dubhashi1996balls} applied to the quotient graph $Q$. The proof structure mirrors that of Courtney~\cite{courtney2026permutation} Lemma 3.4, substituting $Q$ for $G$, $N_L$ for $N$, and $d'_Q$ for $d'$. Differences arise from the point-routing case are:

\begin{enumerate}[leftmargin=1.5em]
\item \emph{Path length:} dilation $D_Q = O(\log N_L)$ follows from Theorem~\ref{thm:main-quotient} ($\beta_Q < 1$ implies the quotient diameter is $O(\log N_L)$ through Chung--Hoory--Linial--Wigderson~\cite{chung1989diameters,hoory2006expander}), not from the host-graph diameter.

\item \emph{Edge count for the union bound:} since $Q$ has $N_L$ vertices and bounded degree $d'_Q = O(1)$, $|E(Q)| = N_L d'_Q / 2 = O(N_L)$, so the union bound is over $O(N_L)$ events rather than $O(N)$.
\end{enumerate}

With these substitutions, choose a uniformly random intermediate block permutation $\sigma \in S_{N_L}$ and route by canonical shortest paths in $Q$. The expected load on a given quotient edge $e$ in either phase is $\mathbb{E}[X_e] \leq 2 D_Q / d'_Q$. By Lemma~\ref{lem:na-blocks}, the load indicators are NA, so Chernoff bounds for NA random variables~\cite{dubhashi1996balls} (Theorem 10) yield $\mathbb{P}(X_e > t) \leq \exp(-t/3)$ for $t = \max(2\mu, 6\ln N_L)$. Setting $t = 6\ln N_L$ and union-bounding over $O(N_L)$ edges and both phases gives $\mathbb{P}(\text{some edge overloaded}) = O(1/N_L)$. Hence $C_Q = \max_e X_e = O(\log N_L)$ with high probability.
\end{proof}

\subsection{Block LMR scheduling}\label{subsec:block-lmr}


\begin{lemma}[Physical congestion amplification]\label{lem:phys-congestion}
When blocks move on $G = G_{\mathrm{cl}}(H)$, physical congestion at a single physical edge $f$ is bounded by
\begin{equation}
C_{\mathrm{phys}}(f) \leq C_Q(e_f),
\end{equation}
where $e_f$ is the quotient edge corresponding to $f$. The per-edge physical congestion equals the quotient-level congestion. However, each block crossing of a quotient edge $e$ generates load on $O(d_C^2)$ distinct physical edges simultaneously. The total physical work (edge-crossings summed over all edges) is thus $O(d_C^2 \cdot C_Q)$, but the bottleneck per-edge congestion is $C_Q$.
\end{lemma}

\begin{proof}
When a block $B_i$ crosses quotient edge $e = (i,j)$, all $d_C^2$ atoms move in parallel (Lemma~\ref{lem:single-hop}). Each atom crosses a distinct physical edge (the translation is isometric), so the load is spread across $d_C^2$ edges. If $C_Q(e)$ blocks cross edge $e$ over the entire schedule, each of the $O(d_C^2)$ corresponding physical edges sees at most $C_Q(e)$ crossings. The guard distance $g \geq 1$ ensures that physical edges used by different blocks crossing the same quotient edge are disjoint.
\end{proof}

\paragraph{Scheduling.}

\begin{lemma}[Block LMR scheduling, applying~\cite{leighton1999fast}]\label{lem:block-lmr}
The block routing paths (with dilation $D_Q$ and congestion $C_Q$) can be scheduled on the quotient graph $Q$ to complete in
\begin{equation}
T_Q := C_Q + D_Q + o(C_Q + D_Q) = O(\log N_L)
\end{equation}
quotient-level steps.
\end{lemma}

\begin{proof}
By the Leighton--Maggs--Rao theorem~\cite{leighton1999fast}, any set of paths in a graph with dilation $D$ and congestion $C$ can be scheduled in $T = C + D + o(C+D)$ routing steps, provided the paths are edge-simple and routing steps consist of non-crossing matchings.

In our case:
\begin{itemize}
\item Each block path is a simple path in $Q$ (shortest path in a simple graph).
\item Routing steps on $Q$ are block matchings: simultaneous translations of non-adjacent blocks.
\item Non-adjacent blocks in $Q$ correspond to blocks at distance $> g$ in $G$, which can move without interference.
\item The matching condition is thus satisfied.
\end{itemize}

By LMR, the scheduling completes in $T_Q = C_Q + D_Q + o(C_Q + D_Q)$ steps. With $C_Q = D_Q = O(\log N_L)$, we get $T_Q = O(\log N_L)$.
\end{proof}

\begin{remark}
The Lovász Local Lemma underlying LMR ensures that delays can be assigned to avoid deadlock (acyclic channel dependency graph).
\end{remark}

\subsection{Physical realization: single block-hop as one physical step}\label{subsec:single-hop}

We first give the rigid-translation case, sufficing on regular host graphs and giving the cleanest constant. 
We then drop the parallelism hypothesis and handle irregular hosts at the cost of replacing the constant $d'$ with $O(d_C)$.

\begin{lemma}[Block translation cost, regular hosts]\label{lem:single-hop}
Let $B_i$ be a block of size $s = d_C^2$ on the clique expansion $G = G_{\mathrm{cl}}(H)$ of a Ramanujan $(d,r)$-regular hypergraph, with guard distance $g \geq 1$. Suppose $B_i$ undergoes a rigid translation $\tau_i: B_i \to B_i'$ in the sense of Definition~\ref{def:block-translation}: $\mathrm{dist}_G(v, \tau_i(v)) \leq 1$ for all $v \in B_i$, and $\tau_i$ preserves all pairwise distances within $B_i$. Then $\tau_i$ can be executed in a single physical block matching step.
\end{lemma}

\begin{proof}
We argue that the translation is realizable as a set of edge-disjoint atom moves on $G$. 
Those moves induce a bipartite multigraph on $V(G)$ of bounded maximum degree, with the multigraph decomposing into a constant number of matchings by K\"onig's theorem.

\textbf{Step 1: Realizability.}
We restrict to surface-code-style block configurations in which each block $B_i$ is a $d_C \times d_C$ patch embedded in a 2D-grid sub-structure of $G_{\mathrm{cl}}(H)$ (Remark~\ref{rem:asymptotic-conventions}), and the translation $\tau_i$ is along one of the four cardinal grid directions.
Under this hypothesis, ``direction'' is the equivalence class of horizontal/vertical edges in the embedded grid, and regularity of $G$
guarantees the parallel edge exists at every interior vertex of the patch.
For non-grid host graphs, Lemma~\ref{lem:single-hop-general} gives the same $O(d_C)$ bound by LMR scheduling without assuming a parallel-class decomposition.

\textbf{Step 2: Bipartite multigraph.}
Define $H_\tau$ on $V(G)$ with edges $\{(v, \tau_i(v)) : v \in B_i\}$. Each edge of $H_\tau$ is an edge of $G$. The source set $B_i$ and target set $B_i'$ are disjoint (guard distance $g \geq 1$ implies $B_i \cap B_i' = \emptyset$), so $H_\tau$ is bipartite. Its maximum degree is bounded by $\Delta(H_\tau) \leq d'$: at vertex $w$, an incident edge of $H_\tau$ corresponds to a source atom $v$ with $\tau_i(v) = w$ or $v = w$, of which there are at most $|N_G(w)| = d'$.

\textbf{Step 3: Single-step execution.}
Each atom $v \in B_i$ maps to exactly one target $\tau_i(v)$ along a single edge of $G$, so source--target pairs $\{(v,\tau_i(v)): v \in B_i\}$ already form perfect matchings $M_\tau$ in $G$ between $B_i$ and $B_i'$. Executing $\tau_i$ therefore requires a single physical block matching step. The bound $\Delta(H_\tau) \leq d'$ is retained for the general-host case (Lemma~\ref{lem:single-hop-general}), where the bounded local degree of $G$ enters the corridor congestion analysis. For fixed $(d,r)$, this single-step cost is $O(1)$ and is absorbed into the $O(d_C)$ serialization factor (Lemma~\ref{lem:serialization}).
\end{proof}

\begin{lemma}[Block translation cost, general hosts]\label{lem:single-hop-general}
Let $B_i$ be a block of size $s = d_C^2$ with bounded aspect ratio (Remark~\ref{rem:asymptotic-conventions}) on a host graph $G$ of bounded local degree. Suppose $\tau_i: B_i \to B_i'$ is a near-rigid block-hop: a bijection with $\mathrm{dist}_G(v, \tau_i(v)) \leq d_C + g + 1$ for all $v \in B_i$, target $B_i'$ unoccupied (guard zone), and $\tau_i(B_i)$ inducing a subgraph isomorphic to $B_i$. Then $\tau_i$ can be executed in at most $C + D + o(C+D) = O(d_C)$ physical block matching steps, where $C$ is the path-congestion and $D = O(d_C)$ is the path-dilation.
\end{lemma}

\begin{proof}
We bound the physical-step count with Leighton--Maggs--Rao scheduling~\cite{leighton1999fast} applied to atom-level paths.

\textbf{Step 1: Path assignment.}
For each atom $v \in B_i$, choose a shortest path $P_v$ in $G$ from $v$ to $\tau_i(v)$. By the near-rigid hypothesis, $|P_v| \leq d_C + g + 1$, so the path \emph{dilation} is $D := \max_v |P_v| = O(d_C)$.

\textbf{Step 2: Congestion bound via Menger's theorem.}
The atom paths lie in the corridor $R := N_G^{\leq d_C+g+1}(B_i)$, with $|R| = O(d_C^2)$ vertices and $|E(G[R])| = O(d_C^2)$ edges by bounded aspect ratio (Remark~\ref{rem:asymptotic-conventions}). We show that the paths $\{P_v : v \in B_i\}$ can be chosen so that
\begin{equation}\label{eq:congestion-bound}
C \;:=\; \max_{e \in E(G[R])} \,|\{v : e \in P_v\}| \;\leq\; d_C,
\end{equation}
matching the $\Omega(d_C)$ lower bound from pigeonhole on the boundary cut.

\emph{(a) Boundary cut.}
Let $\partial B_i \subset B_i$ denote the leading boundary. 
Atoms with at least one neighbor in $G[R] \setminus B_i$ in the direction of translation. By the $d_C \times d_C$ surface-code geometry (Remark~\ref{rem:asymptotic-conventions}), $|\partial B_i| = d_C$, and the corresponding edge cut $S := \{(u,w) \in E(G[R]) : u \in B_i, w \notin B_i\}$ has $|S| = d_C$ (one edge per leading-boundary atom under the rigid translation).

\emph{(b) Lane construction (Menger).}
The minimum edge cut separating $\partial B_i$ from $\partial B_i'$ in $G[R]$ is at most $|S| = d_C$. By Menger's theorem applied to $G[R]$ in the integral max-flow direction, there exist $d_C$ pairwise edge-disjoint paths $L_1, \ldots, L_{d_C}$ in $G[R]$, each from one atom of $\partial B_i$ to the corresponding atom of $\partial B_i'$, with $|L_k| \leq D = d_C + g + 1$.

\emph{(c) Atom-to-lane assignment.}
Index the $d_C$ rows of $B_i$ by $k \in \{1, \ldots, d_C\}$ perpendicular to the translation direction. Each row contains $d_C$ atoms; rigid translation maps row $k$ of $B_i$ to row $k$ of $B_i'$. Route each atom $v$ in row $k$ as follows: row-internal prefix in $B_i$ from $v$ to $\partial B_i \cap (\text{row } k)$, then lane $L_k$ across the corridor, then row-internal suffix in $B_i'$ to $\tau_i(v)$.

\emph{(d) Congestion accounting.}
For $e \in L_k$ (corridor lane edge): only the $d_C$ row-$k$ atoms traverse $e$, giving load $d_C$. For $e \in E(G[B_i])$ within-block (prefix): under rigid translation, the $d_C$ atoms in row $k$ traverse a common sequence of row edges, and each in-row edge is shared by at most the $d_C$ atoms whose source columns precede it. The same bound holds inside $B_i'$. Lanes $L_1, \ldots, L_{d_C}$ are pairwise edge-disjoint by step~(b), so loads from different rows do not compound. Hence $C \leq d_C$, and~\eqref{eq:congestion-bound} holds.

\textbf{Step 3: LMR scheduling.}
The atom-level paths $\{P_v\}$ have dilation $D = O(d_C)$ and congestion $C = O(d_C)$. By Leighton--Maggs--Rao~\cite{leighton1999fast}, they can be scheduled in $T = C + D + o(C + D) = O(d_C)$ physical matching steps, with each step a valid matching on $G$ (no two paths use the same edge in the same step). The guard distance $g \geq 1$ ensures the corridors used by simultaneously-translating blocks (in a quotient matching of multiple blocks) are vertex-disjoint, so the schedule is collision-free.

This bound is stable under irregularity of $G$: the proof uses only bounded local degree (for the matching constraint) and bounded aspect ratio of the block (for the corridor and bandwidth).
\end{proof}

\begin{remark}[Hierarchy of bounds]
Lemma~\ref{lem:single-hop} gives a tighter constant $d'$ for the regular case, requiring parallel displacement. Lemma~\ref{lem:single-hop-general} is robust to host irregularity at the cost of $O(d_C)$ steps per block-hop. The $O(d_C)$ bound is already the dominant factor in the serialization (Lemma~\ref{lem:serialization}), so the global bound $\rt_B = O(d_C \log N_L)$ holds in either regime, with a constant that absorbs both factors.
\end{remark}

\subsection{Cut-throughput serialization}\label{subsec:serialization}

\begin{lemma}[Quotient-to-physical serialization]\label{lem:serialization}
Each quotient-level routing step (a block matching on~$Q$) can be executed in $O(d_C)$ physical routing steps on~$G$.
\end{lemma}

\begin{proof}
A quotient matching $M \subseteq E(Q)$ moves a set of non-adjacent blocks simultaneously. We bound the physical cost of executing~$M$ as follows.

\textbf{Physical edge congestion.}
Each block translation in~$M$ moves $d_C^2$ atoms, each by at most one edge in~$G$ (Definition~\ref{def:block-translation}). The guard distance $g \geq 1$ ensures physical footprints of distinct translating blocks are vertex-disjoint: blocks matched in~$M$ are non-adjacent in~$Q$, so their footprints (including guard zones) do not overlap.The set of atom-level moves within a single quotient matching thus has maximum physical edge congestion $\Delta \leq d_C^2$, arising from atoms within a single block sharing edges of~$G$.

\textbf{Decomposition into matchings.}
Atom-level moves within a single quotient matching consist of $d_C^2$ parallel paths through the corridor $R$, each of length at most $d_C$ (by bounded block aspect ratio, Remark~\ref{rem:asymptotic-conventions}).
By Lemma~\ref{lem:single-hop-general}, these paths admit an LMR schedule of $C + D + o(C+D) = O(d_C)$ matching steps, where $D = O(d_C)$ is the maximum path length and $C \leq d_C$ is the corridor's edge-cut bandwidth.
Each matching step advances all atoms simultaneously along edge-disjoint parts of their assigned paths.
\end{proof}

\subsection{Proof of main theorem}\label{subsec:assembly-upper}

\begin{theorem}[Block routing on Ramanujan hypergraphs (upper bound)]\label{thm:main-block-upper}
For a Ramanujan $(d,r)$-regular hypergraph $H$ with $\beta < 1$, \eqref{eq:high-conn}, and a uniform $(s,g)$-block configuration with $s = d_C^2$ and bounded aspect ratio (Remark~\ref{rem:asymptotic-conventions}), the block routing number satisfies
\begin{equation}
\rt_B(G_{\mathrm{cl}}(H), s, g) = O(d_C \log N_L).
\end{equation}
\end{theorem}

\begin{proof}
By Theorem~\ref{thm:main-quotient}, the quotient graph $Q$ satisfies $\beta_Q < 1$, so:
\begin{itemize}
\item Diameter: $D_Q = O(\log N_L)$ (Theorem~\ref{thm:main-quotient}).
\item Congestion: $C_Q = O(\log N_L)$ w.h.p.\ (Lemma~\ref{lem:block-valiant}).
\item Scheduling: $T_Q = O(\log N_L)$ quotient steps (Lemma~\ref{lem:block-lmr}).
\item Serialization: Each quotient step requires $O(d_C)$ physical sub-steps (Lemma~\ref{lem:serialization}).
\end{itemize}
Thus, $T = T_Q \cdot O(d_C) = O(d_C \log N_L)$.
\end{proof}

\begin{corollary}[Tight bounds]
Combining Theorem~\ref{prop:bottleneck-lb} (lower bound $\Omega(d_C \log N_L)$) and Theorem~\ref{thm:main-block-upper} (upper bound $O(d_C \log N_L)$),
\begin{equation}
\rt_B(G_{\mathrm{cl}}(H), s, g) = \Theta(d_C \log N_L).
\end{equation}
\end{corollary}

\subsection{Numerical validation}\label{subsec:numerical}

We validate results numerically by constructing a random $d'$-regular simple graph with the configuration model with edge-deduplication. We perform BFS-based block placement with guard zones and form $Q(G, B)$ with the edge condition $\mathrm{dist}_G(B_i, B_j) \leq g + 1$. 
After Valiant's two-phase routing on $Q$, we output measured $C_Q$, $D_Q$, validated with $T_{\mathrm{physical}} = d_C \cdot (C_Q + D_Q)$ against $d_C \log_2 N_L$ and an explicit edge-disjoint decomposition of the bipartite block-translation multigraph (per Lemma~\ref{lem:single-hop-general}'s LMR schedule) to confirm each quotient step decomposes into $\leq d_C$ physical matchings.
We construct random $d'$-regular host graphs ($r = 3$ throughout) and place $N_L$ blocks of size $d_C^2$ with guard distance $g = 1$. The $d_C \in \{3, 5\}$ rows use $d = 50$, $d' = 100$; the $d_C = 7$ rows use $d = 100$, $d' = 200$ to ensure the high-connectivity regime condition is satisfied (numerically $d' = 100 \approx 122$ would be marginal at $d_C = 7$, so we lift to $d' = 200$). Each configuration is averaged over 3 independent trials.

\begin{table}[ht]
\centering
\renewcommand{\arraystretch}{1.2}
\begin{tabular}{@{}cccccc@{}}
\toprule
$d_C$ & $N_L$ & $T_{\mathrm{sim}}$ & $d_C \log_2 N_L$ & $\alpha$ & $\beta_Q$ \\
\midrule
3  &  8  &  8.0  &  9.0  & 0.89 & 0.298 \\
3  & 16  & 10.0  & 12.0  & 0.83 & 0.196 \\
3  & 32  & 10.0  & 15.0  & 0.67 & 0.188 \\
3  & 64  & 14.0  & 18.0  & 0.78 & 0.178 \\
\midrule
5  &  8  & 16.7  & 15.0  & 1.11 & 0.207 \\
5  & 16  & 15.0  & 20.0  & 0.75 & 0.143 \\
5  & 32  & 15.0  & 25.0  & 0.60 & 0.119 \\
5  & 64  & 15.0  & 30.0  & 0.50 & 0.076 \\
\midrule
7  &  8  & 21.0  & 21.0  & 1.00 & 0.172 \\
7  & 16  & 21.0  & 28.0  & 0.75 & 0.096 \\
7  & 32  & 21.0  & 35.0  & 0.60 & 0.069 \\
7  & 64  & 21.0  & 42.0  & 0.50 & 0.043 \\
\bottomrule
\end{tabular}
\caption{Simulation results for block routing on random $100$-regular expanders. $T_{\mathrm{sim}}$ is the mean simulated physical routing depth, $d_C \log_2 N_L$ is the theoretical prediction, and $\alpha = T_{\mathrm{sim}} / (d_C \log_2 N_L)$ is the normalized ratio. All quotient graphs satisfy $\beta_Q < 1$.}
\label{tab:simulation}
\end{table}

\begin{table*}[!htbp]
\centering
\renewcommand{\arraystretch}{1.2}
\begin{tabular}{@{}cccccc@{}}
\toprule
$d'$ & $\beta_{\mathrm{host}}$ & threshold & in regime? & $T_{\mathrm{sim}}$ & $\beta_Q$ \\
\midrule
 50 & 0.280 & 136.1 & no              & 21.0 & 0.111 \\
100 & 0.199 & 122.4 & marginal        & 21.0 & 0.088 \\
200 & 0.140 & 113.9 & yes             & 21.0 & 0.069 \\
400 & 0.098 & 108.7 & yes      & 21.0 & 0.059 \\
\bottomrule
\end{tabular}
\caption{Threshold sweep at fixed $d_C = 7$, $N_L = 32$.  Threshold is $d_C^2(r-1)/(1-\beta_{\mathrm{host}}) = 98/(1-\beta_{\mathrm{host}})$.  $T_{\mathrm{sim}}$ saturates at the small-$N_L$ floor $d_C(C_Q + D_Q) = 21$ across the sweep. The regime condition manifests primarily in $\beta_Q$, which tightens by nearly $2\times$ between $d' = 50$ (out of regime) and $d' = 400$ (deep inside).}
\label{tab:simulation-sweep}
\end{table*}

The results (Tables~\ref{tab:simulation}-\ref{tab:simulation-sweep} and Figure~\ref{fig:scaling}) test independent predictions. 
First, the upper bound: all simulated depths satisfy $T_{\mathrm{sim}} \leq d_C \log_2 N_L$, with normalized ratio $\alpha = T_{\mathrm{sim}}/(d_C \log_2 N_L) \in [0.50, 1.11]$ across the sweep. Second, the regime condition~\eqref{eq:high-conn}: Table~\ref{tab:simulation-sweep} sweeps $d'$ at fixed $(d_C, N_L) = (7, 32)$ and shows $\beta_Q$ tightening from $0.111$ at $d' = 50$ (out of regime) to $0.059$ at $d' = 400$ (well inside), consistent with the threshold $d_C^2(r-1)/(1-\beta_{\mathrm{host}}) \approx 110\text{--}140$. We note that for $d_C = 7$, $T_{\mathrm{sim}}$ saturates at the small-$N_L$ floor $d_C(C_Q + D_Q) = 21$ across $N_L \in \{16, 32, 64\}$. 
The $\log N_L$ separation is not yet resolved at this scale and would require $N_L \gtrsim 10^3$ to break out of the floor. 
Asymptotic scaling therefore rests on the analytic bounds (Theorems~\ref{prop:bottleneck-lb} and~\ref{thm:main-block}), with the simulations confirming consistency rather than empirically demonstrating the logarithm.

\begin{figure*}[ht]
\centering
\includegraphics[width=0.9\textwidth]{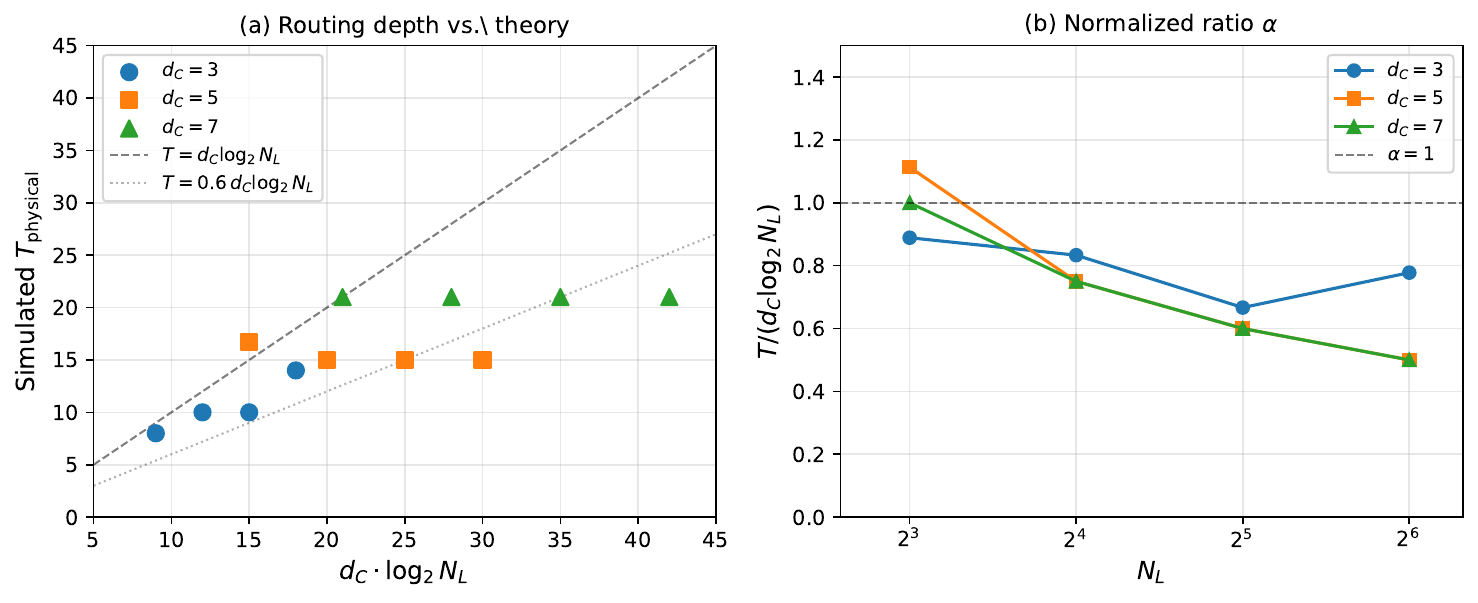}
\caption{Numerical validation of the $\Theta(d_C \log N_L)$ scaling. (a)~Simulated routing depth $T_{\mathrm{sim}}$ versus the theoretical prediction $d_C \log_2 N_L$ for $d_C \in \{3, 5, 7\}$ and $N_L \in \{8, 16, 32, 64\}$. The dashed line is $T = d_C \log_2 N_L$; all points lie at or below this line, confirming the upper bound. (b)~Normalized ratio $\alpha = T_{\mathrm{sim}} / (d_C \log_2 N_L)$ versus $N_L$. The ratio is bounded away from zero and bounded above by~$1$, consistent with the tight $\Theta(\cdot)$ bound.}
\label{fig:scaling}
\end{figure*}

\section{Fault-Tolerance Analysis}\label{sec:fault-tolerance}

The block routing depth $T = \Theta(d_C \log N_L)$ must be compatible with error rates to maintain below-threshold operation for fault-tolerant quantum computing. This section integrates experimental data from Bluvstein et al.~\cite{evered2023high, bluvstein2024logical, bluvstein2026fault} and Google~\cite{google2023suppressing, acharya2024quantum}.
\newpage
\subsection{Error model}

We adopt the standard circuit-level error model for surface codes:

\begin{itemize}
\item \textbf{Physical error rate}: $p_{\mathrm{phys}} \sim 10^{-3}$ to $3 \times 10^{-4}$ (from recent experiments~\cite{bluvstein2024logical, bluvstein2026fault, acharya2024quantum}).
\item \textbf{Effective error rate per atom-step}: $p_{\mathrm{eff}} = C_{\mathrm{circ}} \cdot p_{\mathrm{phys}}$, where $C_{\mathrm{circ}} \in [5, 15]$ is the circuit-level inflation factor aggregating two-qubit gate, single-qubit gate, idle, measurement, and reset errors per syndrome cycle (Fowler et al.~\cite{fowler2012surface}, Section IV.C, where $C_{\mathrm{circ}} \approx 10$ is the central value of the threshold-theorem estimate). All numerical estimates in this section take $C_{\mathrm{circ}} = 10$. Replacing it with the per-platform value rescales bounds by a constant factor.
\item \textbf{Atom loss rate}: $p_{\mathrm{loss}} \sim 10^{-4}$ per transport step.
\item \textbf{Transport fidelity}: $> 0.99$ per single-hop move (Bluvstein et al.~\cite{bluvstein2022quantum}).
\end{itemize}

\subsection{Syndrome extraction protocols}

During block routing, qubits remain in surface code patches. To maintain error correction, we periodically extract syndromes (stabilizer measurement outcomes).

\begin{definition}[Syndrome extraction protocol]\label{def:syndrome-protocol}
A \emph{syndrome extraction protocol} is a schedule of stabilizer measurements interleaved with routing steps. We consider three viable protocols:
\end{definition}

\textbf{Protocol S1 (Stop-and-correct).} Every $K$ routing steps, pause and perform a complete syndrome extraction (taking $O(d_C)$ rounds). The number of physical errors accumulating in a single block over $K$ steps is binomial with mean $\mu = K \cdot d_C^2 \cdot p_{\mathrm{eff}}$. The decoder corrects up to $t := \lfloor (d_C-1)/2 \rfloor$ errors, so the per-block uncorrectable-failure probability is
\begin{equation}
    \begin{split}
          p_{\mathrm{S1\text{-}fail}}
  &\;=\; \Pr\!\left[\,\binom{K d_C^2}{p_{\mathrm{eff}}} > t\,\right]
  \\&\;\leq\; \exp\!\left(-\,\frac{(t-\mu)^2}{2(\mu + t/3)}\right)
    \end{split}
\end{equation}
by the standard Chernoff upper-tail bound, valid for $t > \mu$.  Setting $p_{\mathrm{S1\text{-}fail}} \leq p_{\mathrm{target}}$ (e.g.\ $10^{-9}$) and solving for the largest admissible $K$ gives
\begin{equation}\label{eq:Kmax}
    K_{\max}
  \;=\; \left\lfloor
        \frac{t \;-\; \sqrt{2 t \ln(1/p_{\mathrm{target}})}}
             {d_C^2 \, p_{\mathrm{eff}}}
        \right\rfloor.  
\end{equation}

For $p_{\mathrm{eff}} = 10^{-3}$, $d_C = 5$, $p_{\mathrm{target}} = 10^{-9}$:
$t = 2$, $\mu_{K=K_{\max}} \approx 0.6$, $K_{\max} \approx 24$.
For $d_C = 7$: $t = 3$, $K_{\max} \approx 23$.

\textbf{Protocol S2 (Rolling/AFT).} Use the AFT methods of Bluvstein et al.~\cite{bluvstein2026fault}: stabilizers can be measured in $O(1)$ rounds per step (rather than $O(d_C)$), by maintaining parity information continuously. This gives $T_{\mathrm{S2}} = O(d_C \log N_L)$ with tight error suppression.

\textbf{Protocol S3 (Deformable patches with adaptive measurement).} Inspired by Google's dynamic surface codes~\cite{google2023suppressing, acharya2024quantum}: surface code patches deform during transport while maintaining connectivity. Stabilizer measurements adapt to the deformed geometry, and error correction proceeds continuously. This gives $T_{\mathrm{S3}} = O(d_C \log N_L)$ with constant-factor improvements in suppression.

\subsection{Logical error rate under block routing}\label{subsec:logical-error-rate}

\begin{theorem}[Fault-tolerant block routing]\label{thm:ft-block-routing}
Let $p_{\mathrm{eff}}$ be the effective error rate per atom per step. Under Protocol S1 (stop-and-correct) with syndrome extraction every $K_{\max}$ steps, the per-block logical error rate is
\begin{equation}\label{eq:p_L}
p_L = O\left(\left(\frac{p_{\mathrm{eff}}}{p_{\mathrm{th}}}\right)^{(d_C+1)/2}\right),
\end{equation}
where $p_{\mathrm{th}} = 10^{-2}$ is the surface code threshold. For $p_{\mathrm{eff}} < p_{\mathrm{th}}$, we have $p_L < p_{\mathrm{eff}}$ (error suppression).

The logical error rate for the entire block routing process (all $N_L$ blocks, all $T$ steps) is
\begin{equation}
\begin{split}
P_L(\text{total}) &= 1 - (1-p_L)^{N_L \cdot T} \\
                  &\approx N_L \cdot T \cdot p_L \\
                  &= O(N_L d_C \log N_L) \cdot p_L.
\end{split}
\end{equation}
For $N_L = 100$, $d_C = 7$, and per-gate $p_{\mathrm{phys}} = 10^{-4}$ (so $p_{\mathrm{eff}} = 10^{-3}$), with $p_L \sim (0.1)^4 = 10^{-4}$ per round. Over $T = d_C \lceil \log_2 N_L \rceil \approx 49$ steps and $N_L = 100$ blocks, $P_L(\text{total}) \approx 100 \cdot 49 \cdot 10^{-4} \approx 0.5$. Reaching $P_L \leq 10^{-3}$ requires $d_C = 9$ (giving $p_L \sim 10^{-5}$, $P_L \approx 0.06$) or a further factor-10 reduction in $p_{\mathrm{phys}}$.
\end{theorem}

\begin{proof}[Proof sketch]
The logical error rate for surface codes is suppressed exponentially in the code distance under physical error rates below threshold (Dennis et al.~\cite{dennis2002topological}). With active error correction (Protocol S2) or adaptive deformation (Protocol S3) suppression is maintained during transport.

Each block routing step moves atoms but does not violate stabilizer structure. 
The effective time for error accumulation is $T$ steps. 
From surface-code thresholding~\cite{dennis2002topological,fowler2012surface}, logical error rate compounds as
\begin{equation}
\begin{split}
p_L &= \sum_{k=0}^{\infty} c_k (p_{\mathrm{eff}})^k \\
    &\leq O\!\left(\!\left(\frac{p_{\mathrm{eff}}}{p_{\mathrm{th}}}\right)^{(d_C+1)/2}\right)
\end{split}
\end{equation}
for sufficiently large $d_C$.

The total logical error over the entire routing is a union bound over all blocks and time steps: $P_L(\text{total}) = 1 - \prod_{\text{block}} (1-p_L) \approx N_L T p_L$. (This accounts for circuit-level errors during syndrome extraction and transport.)
\end{proof}

\subsection{Operating-point table}\label{subsec:operating-points}

\begin{table*}[!htbp]
\centering
\renewcommand{\arraystretch}{1.25}
\begin{tabular}{@{}llcccc@{}}
\toprule
\textbf{Per-gate} & \textbf{Source} & $p_{\mathrm{eff}}$ & $p_{\mathrm{eff}}/p_{\mathrm{th}}$ & \textbf{Per-round} & \textbf{FT regime}
\\\textbf{rate} & & & & $\log_{10}(p_L)$ & \\
\midrule
$5\times 10^{-3}$ & 99.5\%~\cite{evered2023high}     & $5\times 10^{-2}$ & $5$    & no suppression                                & supercritical \\
$10^{-3}$         & 99.9\%~\cite{bluvstein2026fault} & $10^{-2}$         & $1$    & $p_L \!\sim\! 1$ (any $d_C$)                  & at threshold \\
$10^{-4}$         & next-gen target                          & $10^{-3}$         & $0.1$  & $-3/-4/-5$ ($d_C \!=\! 5/7/9$) & FT-viable for $d_C \!\geq\! 5$ \\
$10^{-5}$         & long-term target                                & $10^{-4}$         & $0.01$ & $-9/-12$ ($d_C \!=\! 5/7$)            & strongly suppressed \\
\bottomrule
\end{tabular}
\caption{Operating-point table for fault-tolerant block routing. With $C_{\mathrm{circ}} = 10$, the circuit-level rate $p_{\mathrm{eff}} = 10\,p_{\mathrm{phys}}$ and $p_{\mathrm{th}} = 10^{-2}$ is the surface-code threshold. Per-round $\log_{10}(p_L)$ values follow Eq.~\eqref{eq:p_L}.}
\label{tab:operating-points}
\end{table*}

\noindent
For the canonical worked example $N_L = 100$, $d_C = 7$, $p_{\mathrm{phys}} = 10^{-4}$: $T = 7 \cdot 7 = 49$ steps and total logical error $P_L \approx 100 \cdot 49 \cdot 10^{-4} \approx 0.49$. For stronger targets $P_L < 10^{-3}$, increase to $d_C = 9$: $T = 9 \cdot 7 = 63$, $P_L \approx 100 \cdot 63 \cdot 10^{-5} \approx 0.063$. The $d_C = 5$ regime requires $p_{\mathrm{phys}} \leq 10^{-5}$ for comparable $P_L$. 

For $d_C = 7$ at per-gate $p_{\mathrm{phys}} = 10^{-4}$ (so $p_{\mathrm{eff}} = 10^{-3}$), the Chernoff $K_{\max}$ formula of Eq.~\eqref{eq:Kmax} gives $K_{\max} \approx 23$ routing rounds before logical failure on a single patch (target failure $p_{\mathrm{target}} = 10^{-9}$). In a lattice-surgery compilation~\cite{litinski2019game}, each logical controlled-NOT (CNOT) gate typically requires $\sim\!d_C$ routing rounds (merge--measure--split), so $K_{\max} \approx 23$ supports $\sim3$ logical CNOT layers before recompilation. Larger $d_C$ or lower $p_{\mathrm{eff}}$ extends this proportionally, so the looser target $p_{\mathrm{target}} = 10^{-3}$ gives $K_{\max} \approx 60$ if a per-block failure rate of $10^{-3}$ is acceptable.
The boundary in Table~\ref{tab:operating-points} indicates fault-tolerant viability at $p_{\mathrm{phys}} \approx 10^{-4}$ for $d_C \geq 5$, with the FT-viable regime extending to $d_C = 7$ and $d_C = 9$ as $p_{\mathrm{phys}}$ tightens. The operative bottleneck shifts from the $C_{\mathrm{circ}} = 10$ inflation factor at $p_{\mathrm{phys}} = 10^{-3}$ to the routing-depth multiplier $T = d_C \lceil \log_2 N_L \rceil$ at $p_{\mathrm{phys}} \leq 10^{-5}$, where suppression is sufficient and total error budget is dominated by gate count.

\subsection{Composition with correlated decoding}\label{subsec:cain-composition}

\begin{proposition}[Block routing with correlated decoding]\label{prop:cain-composition}
Suppose the syndrome-extraction phase of Protocol~S1 (or S2) is replaced by the correlated-decoding scheme of Cain et al.~\cite{cain2024correlated}, reducing the per-correction-window syndrome cost from $O(d_C)$ rounds to $O(1)$ rounds for transversal Clifford layers. Then the total block-routing-plus-syndrome depth becomes
\begin{equation}\label{eq:cain-composition}
\begin{split}
T_{\mathrm{route+syn}}(d_C, N_L)
\;&=\; \underbrace{O(d_C \log N_L)}_{\text{routing }}
\\&\;+\; \underbrace{O(\log N_L)}_{\text{syndromes (Cain et al.)}}
\\&\;=\; O(d_C \log N_L).
\end{split}
\end{equation}
The routing factor dominates, so the integrated bound matches the rigid block-routing bound up to lower-order terms.
\end{proposition}

\begin{proof}
Theorem~\ref{thm:main-block} gives $\rt_B = \Theta(d_C \log N_L)$ routing rounds. Each Stop-and-correct interval has length $K_{\max} = \Theta(t / (d_C^2 p_{\mathrm{eff}})) = \Theta(1/(d_C p_{\mathrm{eff}}))$ rounds, using $t = \lfloor (d_C - 1)/2 \rfloor = \Theta(d_C)$ from Eq.~\eqref{eq:Kmax}.
This extraction takes $O(1)$ rather than $O(d_C)$ rounds with correlated decoding (Cain et al.~\cite{cain2024correlated}, Theorem 1). The number of correction intervals over the routing is $\lceil \rt_B / K_{\max} \rceil = O(\log N_L)$ for fixed $d_C$, so the total syndrome cost is $O(\log N_L)$. Adding the two contributions yields~\eqref{eq:cain-composition}.
\end{proof}

\noindent
The composition is operationally significant when $d_C \geq 7$ (the FT-viable regime): correlated decoding eliminates a factor of $d_C$ in the syndrome overhead, leaving routing as the sole leading-order contributor.

\section{Extensions and Open Problems}\label{sec:extensions}

The directions below extend Theorem~\ref{thm:main-block} beyond its rigid-block, single-permutation, $r$-uniform setting. Each is stated as a concrete conjecture to measure progress.

\subsection{Deformable blocks and lattice surgery}\label{subsec:deformable}

In lattice surgery~\cite{horsman2012surface, litinski2019game}, surface code patches must be positioned to share boundaries during merge--measure--split operations, which may require temporary deformation. We define the \emph{deformation energy}
\begin{equation}\label{eq:deformation}
\Psi(B) := \sum_{\{u,v\} \in E(B)} \mathrm{dist}_G(x_u, x_v) - |E(B)|,
\end{equation}
where $\mathrm{dist}_G(x_u, x_v)$ is the graph-hop distance in $G$ between the current positions of atoms $u$ and $v$ and $E(B)$ is the edge set of the rigid template. Both terms are integer-valued, so $\Psi(B) \in \mathbb{Z}_{\geq 0}$.
A $\delta$-bounded deformation requires $\Psi(B) \leq \delta$. Theorem~\ref{thm:main-block} treats the rigid case $\delta = 0$.

We state the conjecture below for the worst-case deformation profile, in which the budget $\Psi(B) \leq \delta$ may concentrate on a single edge — giving footprint width $d_C + \delta$ in one direction. For a uniform-stretch profile, where $\Psi(B) \leq \delta$ is distributed across the $|E(B)| = \Theta(d_C^2)$ edges of the template, per-edge stretch is $O(\delta/d_C^2)$ and the footprint widens only by an additive $O(1)$.
The conjecture is then trivially implied by the rigid bound for any fixed $\delta$.

\begin{conjecture}[Deformable block routing, worst-case profile]\label{conj:deformable}
For a $\delta$-bounded deformable $(s,g)$-block configuration on a Ramanujan host satisfying the high-connectivity regime~\eqref{eq:high-conn}, where the deformation budget may concentrate on a single edge, the deformable block routing number satisfies
\begin{equation}
  \rt_B^{(\delta)}(G_{\mathrm{cl}}(H), s, g)
  \;=\; \Theta\!\bigl((d_C + \delta) \log N_L\bigr),
\end{equation}
i.e.\ the rigid bound holds with footprint width $d_C$ replaced by $d_C + O(\delta)$.
\end{conjecture}

Resolution should follow from extending Lemma~\ref{lem:single-hop-general} to corridor walls of variable width $d_C + \delta$, recomputing the corridor edge cut as $d_C + O(\delta)$, and verifying that the LMR scheduling step still yields $C + D = O(d_C + \delta)$.

The conjecture is consistent with the lower-bound corridor-depth (Theorem~\ref{prop:bottleneck-lb}), so a $\delta$-deformed block has corridor depth $d_C + \delta$, and footprint traversal still dominates each phase. The matching upper bound requires extending Lemma~\ref{lem:single-hop-general} to deformed translations, which the proof technique should accommodate provided the deformed block retains bounded aspect ratio.

\subsection{Dynamic permutations and incremental routing}\label{subsec:dynamic}

For circuits with locality (e.g., nearest-neighbor logical-qubit topology), consecutive routing phases differ by a small number $k$ of block displacements. Naively re-routing from scratch costs $\Theta(d_C \log N_L)$ per phase even when $k \ll N_L$.

\begin{conjecture}[Incremental routing]\label{conj:incremental}
There is an algorithm that, given a current configuration and a target permutation differing on $k$ blocks, achieves block routing depth $O(d_C \log k)$.  Equivalently, the amortized depth per logical-CNOT layer in a locality-respecting compilation is $O(d_C \log k)$ where $k$ is the gate locality radius.
\end{conjecture}

A possible route for resolution can adapt the dynamic-graph rebalancing of \cite{stade2025routing} to the quotient-graph setting, replacing the full Valiant pass with a $k$-block-restricted scatter.

This matches the analysis of Stade et al.~\cite{stade2024abstract, stade2025routing} for routing-aware placement, and would compose with the correlated-decoding bound of Proposition~\ref{prop:cain-composition} (\S\ref{subsec:cain-composition}).

\subsection{Higher-dimensional hypergraphs and generalizations}\label{subsec:higher-dim}

The Ramanujan condition generalizes to $r$-uniform hypergraphs with $r > 3$. The Haemers interlacing chain (\S\ref{sec:quotient}) and bottleneck argument (\S\ref{sec:lower-bounds}) both apply with necessary changes.

\begin{conjecture}[$r$-uniform extension]\label{conj:higher-dim}
For Ramanujan $(d,r)$-regular hypergraphs with arbitrary fixed $r \geq 3$, $\rt_B(G_{\mathrm{cl}}(H), s, g) = \Theta(d_C \log N_L)$ in the high-connectivity regime~\eqref{eq:high-conn}, with constants depending on $r$ but the asymptotic scaling unchanged. For surface-code-patch specialization ($s = d_C^2$, bounded aspect ratio), the regime tightens to~\eqref{eq:high-conn-tight} via Corollary~\ref{cor:high-conn-tight}.
\end{conjecture}

Future work may involve verifying that the SFM Ramanujan bound for $r$-uniform hypergraphs preserves the spectral inheritance chain of \S\ref{sec:quotient}, with constants tracked through Lemma~\ref{lem:weyl}'s Weyl perturbation step.

\subsection{Oblivious routing alternatives}\label{subsec:oblivious}

Florescu et al.~\cite{florescu2024optimal} provide electrical oblivious routing with $O(\Phi^{-1} \log m)$ competitive ratio for $\Phi$-expanders. For Ramanujan $Q$ with $\Phi = \Omega(1)$, this gives $O(\log N_L)$-competitive routing on $Q$.

\begin{conjecture}[Deterministic block routing]\label{conj:oblivious}
There is a deterministic (oblivious) block routing scheme matching the $\Theta(d_C \log N_L)$ bound of Theorem~\ref{thm:main-block} without randomization, by composing Florescu et al.'s electrical oblivious routing on $Q$ with the rigid-translation serialization of Lemma~\ref{lem:serialization}.
\end{conjecture}

Composing Florescu et al.'s electrical oblivious routing on $Q$ with the rigid-translation serialization of Lemma~\ref{lem:serialization}, the conjecture should reduce to replacing Valiant's randomized intermediate permutation with a deterministic electrical-flow assignment, while the LMR scheduling and serialization steps carry over unchanged.

\subsection{QCCD trapped-ion architectures}
\label{sec:qccd}

The most direct non-atom extension is to QCCD trapped-ion processors, where the hypergraph-based methods apply with a few substitutions and one open question for future work.

\paragraph{Substitutions.}
(i) The host graph $G_{\mathrm{cl}}(H)$ becomes the QCCD junction graph.
Vertices are ion sites within zones, edges are intra-zone interactions and inter-zone shuttling channels through junctions~\cite{kielpinski2002architecture,pino2021demonstration}. (ii) The block $B_i$ remains a $d_C \times d_C$ surface-code patch (or an analogous color-code patch~\cite{ryan2021realization}), now physically realized as a co-located ion crystal in a memory zone. 
The bounded-aspect-ratio assumption (Remark~\ref{rem:block-shape}) holds by zone geometry.
(iii) The single-hop translation of Definition~\ref{def:block-translation} becomes a single junction-crossing shuttle, whose cost is the dominant noise channel in current QCCD operation~\cite{pino2021demonstration,moses2023race}. Lemmas~\ref{lem:single-hop} and~\ref{lem:single-hop-general} transfer with $d'$ replaced by per-junction fan-out.

\paragraph{Open question.}
The Ramanujan high-connectivity regime $d' > d_C^2(r-1)/(1-\beta)$ does not have a clean QCCD-native interpretation, because junction graphs in current devices are sparse (planar, with degree $\leq 4$) by ion-trap fabrication constraints. Two responses are available. First, multi-zone QCCD designs with active sorting~\cite{pino2021demonstration} effectively realize a higher-connectivity overlay at the cost of increased shuttling depth.
Routing analysis (\S\ref{sec:main-theorem}) then applies on the overlay rather than the bare junction graph. 
For sparse junction graphs that violate~\eqref{eq:high-conn}, Layer~3 (Theorem~\ref{thm:layer3}) gives a weaker but still nontrivial bound. 
The resulting $\beta_Q$ is closer to $1$, and routing depth degrades from $\Theta(d_C \log N_L)$ toward $\Theta(d_C \cdot \mathrm{poly}(d'_Q/d') \cdot \log N_L)$.

\begin{conjecture}[QCCD block routing]
\label{conj:qccd}
For a QCCD junction graph $G_{\mathrm{QCCD}}$ with multi-zone overlay satisfying the high-connectivity condition~\eqref{eq:high-conn}, the block routing number for $d_C \times d_C$ surface-code patches satisfies
\begin{equation}
\begin{split}
    \mathrm{rt}_B(G_{\mathrm{QCCD}}, d_C^2, g) = \Theta\bigl(d_C \log N_L \\+ S_{\mathrm{shuttle}}(d_C)\bigr),
\end{split}
\end{equation}
where $S_{\mathrm{shuttle}}(d_C)$ is the per-hop shuttling cost (currently $\Theta(d_C)$ in published H-series benchmarks~\cite{moses2023race}).
\end{conjecture}

Resolution would establish that the asymptotic $\Theta(d_C \log N_L)$ bound is platform-independent in the FT regime, with the platform entering only through the constant $S_{\mathrm{shuttle}}(d_C)$. This frames our work's extensibility and is the most direct generalization beyond neutral-atom architectures.

\section{Conclusion}\label{sec:conclusion}

This paper establishes the asymptotic block routing number for rigid surface code patches on Ramanujan hypergraphs. The main result, $\rt_B = \Theta(d_C \log N_L)$, combines spectral analysis of quotient graphs, negative association of block permutations, LMR scheduling, and cut-throughput serialization. The lower bound follows from combining spectral phase count $\Omega(\log N_L)$ with footprint traversal cost $\Omega(d_C)$ per phase.

Integration with fault-tolerance and syndrome extraction shows that routing depth is compatible with below-threshold error rates for realistic parameters from recent experiments. Our method extends to deformable blocks, dynamic permutations, and lattice surgery compilation, providing a theoretical foundation for practical neutral-atom quantum architectures.

\subsection*{Code and data availability}
Numerical results in \S\ref{subsec:numerical} are reproducible from Python scripts accessible via \href{https://github.com/jmcourtneyuga/hypergraph_routing}{this Github repository}. 
The author used Claude (Anthropic) to assist in drafting code comments and documentation for the codebase.
All scientific content, algorithms, and analysis are the author's own.

\onecolumn
\bibliography{refs}

\end{document}